\documentclass[aps,reprint,prx]{revtex4-2}
\usepackage{graphicx}% Include figure files
\usepackage{bm}% bold math
\usepackage{color}
\usepackage{amssymb}
\usepackage{mathtools}
\usepackage[mathlines]{lineno}
\usepackage{amsthm}
\newtheorem{proposition}{Proposition}
\usepackage[T1]{fontenc}
\usepackage{textcomp}
\usepackage{gensymb}
\usepackage{silence}
\usepackage{hyperref}

\begin{document}

\preprint{APS/123-QED}

\title{\textbf{Precursors of extreme events and critical transitions} 
}% 

\author{Riccardo Consonni}
\author{Luca Magri}%
 \email{l.magri@imperial.ac.uk; luca\_magri@polito.it}
\affiliation{%
DIMEAS, Politecnico di Torino, Torino 24 10129, Italy
}%
\affiliation{Department of Aeronautics, Imperial College London, London SW7 2AZ, UK.}

\date{\today}% It is always \today, today,

\begin{abstract}
We propose a theory based on dynamical systems to explain and predict the occurrence of extreme events, of which critical transitions form a subset. In fast–slow nonlinear systems, we identify a cascade of events preceding extreme events:
(i) a {\it slow regime}, in which the fast covariant Lyapunov vectors (CLVs) are both tangent to the fast eigenvectors and remain transversal to the slow subspace;
(ii) a {\it transition regime}, in which the fast eigenvalues become neutrally stable while the fast CLVs are no longer tangent to the fast eigenvectors; and
(iii) a {\it critical regime}, in which a strong spectral gap in the eigenvalues causes both fast and slow CLVs to become tangent along the dominant fast direction, breaking the transversality between fast and slow subspaces.
Building on this cascade, we propose two precursors to forewarn the occurrence of extreme events. We numerically test the theory and precursors on low- and higher-dimensional systems. The proposed precursors predict extreme events and critical transitions with 100\% precision and recall. This work opens  opportunities for time-forecasting extreme events  using theoretically grounded precursors. 
\end{abstract}

\maketitle

%\tableofcontents

\section{\label{into}Introduction}

Extreme events \cite{farazmandExtremeEventsMechanisms2019, ghilExtremeEventsDynamics2011} are characterized by large, sudden fluctuations in the state or observables of a dynamical system, which arise in a variety of natural and engineering sectors, such as extreme weather phenomena \cite{ropelewskiGlobalRegionalScale1987, easterlingObservedVariabilityTrends2000}, rogue waves in the ocean and in optical systems \cite{dystheOceanicRogueWaves2008, solliOpticalRogueWaves2007, sapsisStatisticsExtremeEvents2021}, energy dissipation events in turbulent flows \cite{yeungExtremeEventsComputational2015, bloniganAreExtremeDissipation2019}, large-scale power grid failures \cite{crucittiModelCascadingFailures2004, fangSmartGridNew2012}, thermo-acoustic instabilities \cite{huhnStabilitySensitivityOptimisation2020}, climate tipping points and critical transitions \cite{schefferEarlywarningSignalsCritical2009}, among others. In systems with multiple timescales, extreme events are defined as locally unstable regions of the invariant slow manifold, over which the dynamics take place \cite{schefferEarlywarningSignalsCritical2009, schefferAnticipatingCriticalTransitions2012, kuehnMathematicalFrameworkCritical2011, kuehnMultipleTimeScale2015, farazmandExtremeEventsMechanisms2019}. Because of the potentially catastrophic consequences of extreme events, their prediction and mitigation are the subject of both theoretical and data-driven studies \cite[e.g.,][]{dakosSlowingEarlyWarning2008, qiUsingMachineLearning2020, doan2021short, raccaDatadrivenPredictionControl2022a, farandaDynamicalProxiesNorth2017}.

A wide range of methodologies has been developed to diagnose {\it a posteriori} extreme events, including Proper Orthogonal Decomposition (POD) based methods \cite{martinConditionalPODPredicting2025, schmidtConditionalSpaceTime2019}, wavelet decomposition methods \cite{madhusudananPredictingBurstEvents2025}, variational approaches that uncover optimal initial conditions leading to extreme events \cite{farazmandVariationalApproachProbing2017, bloniganAreExtremeDissipation2019}, and machine learning techniques that forewarn the occurrence of extreme events from data \cite[e.g.,][]{doan2021short, raccaDatadrivenPredictionControl2022a, pickeringDiscoveringForecastingExtreme2022, jiangPredictingExtremeEvents2022, foxPredictingExtremeEvents2023}.
Purely empirical observations suggest that extreme events can be forewarned by tracking covariant Lyapunov vectors (CLVs) \cite{beimsAlignmentLyapunovVectors2016, sharafiCriticalTransitionsPerturbation2017}. CLVs \cite{ruelleErgodicTheoryDifferentiable1979, ginelliCovariantLyapunovVectors2013} define the directions of infinitesimal perturbation growth and decay along a trajectory, providing a geometric characterization of the tangent space. Because CLVs evolve covariantly with the dynamics, they capture how infinitesimal perturbations continuously stretch and rotate in time \cite{kuptsovTheoryComputationCovariant2012}. The corresponding Lyapunov exponents quantify the average exponential rates of expansion (or contraction), forming the Lyapunov spectrum that measures the system's degree of chaos and predictability \cite{pikovskyLyapunovExponentsTool2016}. The Lyapunov spectrum and CLVs offer a description of the linear properties of an attractor for the analysis of the behaviour of chaotic systems and their inertial manifolds \cite{yangHyperbolicityEffectiveDimension2009a, yangCovariantLyapunovVectors2012, dingEstimatingDimensionInertial2016}. The angles between CLVs provide a quantitative measure of the system's hyperbolicity and allow the identification of regions where stable and unstable manifolds become nearly tangent, indicating a potential loss of structural stability \cite{kuptsovLyapunovAnalysisStrange2018}. The computation of CLVs has become feasible only in recent years, thanks to the development of efficient numerical algorithms \cite{ginelliCharacterizingDynamicsCovariant2007, wolfeEfficientMethodRecovering2007, froylandComputingCovariantLyapunov2013}. CLVs have also been used to analyse the stability of chaotic systems arising in engineering applications, such as nonlinear oscillations in thermo-acoustics \cite{huhnStabilitySensitivityOptimisation2020, moreira_huhn_2021}, stability of dynamical models learned from data \cite{margazoglouStabilityAnalysisChaotic2023, ozalpReconstructionForecastingStability2023, ozalpStabilityAnalysisChaotic2025, ahmedRobustQuantumReservoir2025}, and systems characterized by multiple timescales, where distinct slow and fast modes interact \cite{eckmannHydrodynamicLyapunovModes2000, carluLyapunovAnalysisMultiscale2019}.
Beyond model-based analyses, several studies have attempted to reconstruct or approximate CLVs directly from data \cite{martinEstimatingCovariantLyapunov2022, brugnagoMachineLearningAlignment2020, viennetGuidelinesDatadrivenApproaches2022}, enabling the application of these geometric indicators to experimental or observational systems. In parallel, other approaches have focused on broader tangent-space diagnostics, for example, optimally time-dependent (OTD) modes and backward Lyapunov vectors \cite{babaeeMinimizationPrincipleDescription2016, farazmandDynamicalIndicatorsPrediction2016, blanchardAnalyticalDescriptionOptimally2019}, which similarly capture directions of transient instability and have shown promise in anticipating the onset of rare or extreme events.

Although tangencies between CLVs have been empirically observed as potential precursors to extreme events \cite{beimsAlignmentLyapunovVectors2016, sharafiCriticalTransitionsPerturbation2017}, a theoretical explanation of the route to extreme events and critical transitions is missing. In particular, existing results on dominated Oseledets splitting describe the regimes in which tangencies do not occur \cite{pughStableErgodicity2004, bochiLyapunovExponentsGeneric2005}, but they do not explain the dynamical processes leading to CLV tangencies. The connection between such mechanisms, extreme events, and invariant manifolds is still not well understood \cite{yangHyperbolicityEffectiveDimension2009a, yangCovariantLyapunovVectors2012, dingEstimatingDimensionInertial2016, carluLyapunovAnalysisMultiscale2019}.

The overarching goal of this paper is to develop a theory that provides a mechanistic explanation of precursors of extreme events in fast-slow systems. In Section \ref{section:setup} we introduce some background from dynamical systems theory. In Section \ref{section:clv_equation}, we introduce the CLV equation and give a geometric interpretation. In Section \ref{sec:theory_roadmap} we provide a roadmap of the theory. In Section \ref{section: stability of clvs}, we calculate the invariant sets of the CLV equation and their stability. In Section \ref{section:fast-slow systems}, we describe the tangent space dynamics of fast-slow systems. In Section \ref{sec: adiabatic condition and dynamical regimes}, we explain the universal route and cascade of mechanisms that occur before extreme events and critical transitions. In Section \ref{section:numerical results}, we test our theory with numerical simulations of a bistable Rössler system, FitzHugh-Nagumo coupled cells, and a variant of the Lorenz-96 model. Section \ref{section:conclusion} concludes the work.

\section{\label{section:setup}Background}
We consider a nonlinear dynamical system

\begin{equation} \label{dynamical_system}
\dot{\mathbf{z}} = \mathbf{F}(\mathbf{z}),
\end{equation}

where $\mathbf{z} \in \mathbb{R}^d$ is the $d$-dimensional state vector defined in the time domain $t \in \mathbb{R}^+$, and $\mathbf{F}: \mathbb{R}^d \rightarrow \mathbb{R}^d$ is a $C^r$, $r \geq 1$ smooth vector field. Given an initial condition $\mathbf{z}(t_0) = \mathbf{z}_0$, the dynamics of \eqref{dynamical_system} are governed by the flow map $\mathbf{z}(t; t_0, \mathbf{z}_0) = \mathbf{\Phi}_{t_0}^{t}(\mathbf{z}_0)$ \cite{arnoldOrdinaryDifferentialEquations1992}. On the other hand, the dynamics of infinitesimal perturbations, $\mathbf{w}(t)$, applied to a trajectory $\bar{\mathbf{z}}(t)$, i.e., $\mathbf{z}(t) = \bar{\mathbf{z}}(t) + \mathbf{w}(t)$, are governed by the Jacobian equation

\begin{equation} \label{eq_2}
\dot{\mathbf{w}} = \mathbf{J}(t) \mathbf{w},
\end{equation}

where $\mathbf{J}(t) = \nabla_{\bar{\mathbf{z}}}\mathbf{F}(\mathbf{z}(t; t_0, \mathbf{z}_0)) \in \mathbb{R}^{d \times d}$ is the Jacobian, which is assumed to be diagonalizable. Given an initial condition $\mathbf{w}(t_0) = \mathbf{w}_0$, the solution of Equation \eqref{eq_2} is $\mathbf{w}(t; t_0, \mathbf{w}_0) = \nabla_\mathbf{z}\mathbf{\Phi}_{t_0}^t(\mathbf{z}_0) \mathbf{w}_0$, where $\nabla_{\mathbf{z}}\mathbf{\Phi}_{t_0}^{t_0}(\mathbf{z}_0) = \mathbf{I}$ and

\begin{align}
\nabla_{\mathbf{z}}\mathbf{\Phi}_{t_0}^{t}(\mathbf{z}_0) \nonumber 
&= \mathcal{T}\exp\left( \int_{t_0}^t \mathbf{J}(\tau) , d\tau \right) \nonumber \\ & = \lim_{N \to \infty}e^{\mathbf{J}(t_N)\Delta t} e^{\mathbf{J}(t_{N-1})\Delta t} \cdots e^{\mathbf{J}(t_0)\Delta t},
\label{eq:timeorderedexponential}
\end{align}

is the tangent propagator, in which the time instants $t_i$ are defined as $t_i = i \Delta t$ for $i = 0,...,N$, and $\Delta t = t / N$, and $\mathcal{T}$ is the time-ordering operator.

\subsection{\label{sec:lyapunov_exponents}Lyapunov Exponents}
Given a vector $\mathbf{w}_0$ and an initial condition $\mathbf{z}_0$, the Lyapunov exponent (LE) is \cite{pikovskyLyapunovExponentsTool2016}

\begin{equation} \label{eq_4}
\bar\lambda(\mathbf{z}_0,\mathbf{w}_0) = \lim_{t \to \infty} \frac{1}{t} \ln\left(\frac{||\nabla_{\mathbf{z}}\mathbf{\Phi}_{t_0}^{t}(\mathbf{z}_0)\mathbf{w}_0||}{||\mathbf{w}_0||}\right).
\end{equation}

The LEs are the asymptotic average exponential rates of expansion of the vector $\mathbf{w}_0$ under the dynamics starting at $\mathbf{z}_0$.

As shown by Oseledets \cite{oseledets1968multiplicative}, there exist $d$ Lyapunov exponents (LEs)
$ \bar \lambda_1(\mathbf{z}_0) \ge \bar \lambda_2(\mathbf{z}_0) \ge \cdots \ge \bar \lambda_d(\mathbf{z}_0),
$ which collectively form the Lyapunov spectrum.
In ergodic systems, the Lyapunov exponents are measures of the attractor, which means that they depend on neither $\mathbf{z}_0$ nor $\mathbf{w}_0$.

The instantaneous exponential rate can be calculated with the instantaneous Lyapunov exponents (ILE) \cite{kuptsovLyapunovAnalysisStrange2018}

\begin{align} \label{eq_6}
\lambda(\mathbf{z},\mathbf{w}) = \frac{{\mathbf{w}^T\mathbf{J}(\mathbf{z})\mathbf{w}}}{{\mathbf{w}^T \mathbf{w}}}.
\end{align}

The Lyapunov exponents are the asymptotic average of the instantaneous Lyapunov exponents. When the vector $\mathbf{w}$ used to compute the ILE in Equation \eqref{eq_6} is the $i$th covariant Lyapunov vector, $\mathbf{u}_i$ (see Section \ref{sec:covariant_lyapunov_vectors}), the corresponding ILE is referred to as the $i$th instantaneous covariant Lyapunov exponent (ICLE).

\subsection{Covariant Lyapunov Vectors}\label{sec:covariant_lyapunov_vectors}
In ergodic systems, the matrices

\begin{equation}\label{eq_7}
\lim_{t \to \pm \infty} \left[\left(\nabla_{\mathbf{z}}\mathbf{\Phi}_{t_0}^{t}\right)^T \nabla_{\mathbf{z}}\mathbf{\Phi}_{t_0}^{t} \right]^{\frac{1}{2t}} = \mathbf{\Xi}^\pm
\end{equation}

exist and are referred to as the forward and backward Oseledets matrices, respectively. The Oseledets matrices are $d$-dimensional square matrices with positive eigenvalues $\bar\gamma_1 \geq ... \geq \bar\gamma_d$ such that $\bar\lambda_i = \ln\bar\gamma_i$ \cite{ginelliCovariantLyapunovVectors2013, kuptsovFastNumericalTest2012}.
The eigenvectors of the forward (backward) matrix are the forward (backward) Lyapunov vectors, which are orthonormal bases of the tangent space.
The intersection of the backward and forward Lyapunov bases \cite{ginelliCovariantLyapunovVectors2013} yields the Oseledets splitting,
$\mathbb{R}^d = E_1 \oplus E_2 \oplus \cdots \oplus E_d$, in which a vector that belongs to $E_i$ grows asymptotically with the corresponding Lyapunov exponent $\bar\lambda_i$

\begin{equation} \label{eq_8}
\bar\lambda_i = \lim_{t \to \pm \infty} \frac{1}{t} \ln \left(\frac{||\mathbf{e}(t)||}{||\mathbf{e}(0)||} \right) \quad \textrm{if} \quad \mathbf{e}(0) \in E_i.
\end{equation}

The Oseledets splitting is covariant with the dynamics, i.e.,
$\mathbf{e}(t_0)$ $\in$ $E_i(\mathbf{z}(t_0))$, $\mathbf{e}(t)$ $=$ $\nabla_{\mathbf{z}}\mathbf{\Phi}_{t_0}^{t}\mathbf{e}(t_0)$ $\Rightarrow \mathbf{e}(t)$ $\in$ $E_i(\mathbf{z}(t)).$
A vector that belongs to the Oseledets subspace $E_i$ is referred to as the $i$th covariant Lyapunov vector (CLV). To numerically compute the CLVs, we employ the algorithm of \cite{ginelliCharacterizingDynamicsCovariant2007, ginelliCovariantLyapunovVectors2013}.

\section{Covariant Lyapunov Vectors: Equation and Geometric Interpretation} \label{section:clv_equation}

We derive an equation for the time evolution of CLVs, which is key to identifying the conditions under which CLV tangencies occur. In Section~\ref{sec: adiabatic condition and dynamical regimes}, we link these conditions to the occurrence of critical transitions and extreme events. A generic CLV with unit length, $\mathbf{u}(t)$\footnote{In this section, we drop the subscript $i$ for brevity.}, is computed by evolving and normalizing a CLV with arbitrary norm, $\hat{\mathbf{u}}(t)$, as \cite{blanchardAnalyticalDescriptionOptimally2019}

\begin{equation} \label{eq_23}
{\mathbf{u}}(t+\delta t) = \frac{\hat{\mathbf{u}}(t+\delta t)}{||\hat{\mathbf{u}}(t + \delta t)||} \quad \forall \delta t.
\end{equation}

Taylor expanding \eqref{eq_23} at first order yields

\begin{align} \label{eq_clv projeciton form}
\dot{{\mathbf{u}}} &= \mathbf{J}(t){\mathbf{u}}(t) - \left({{\mathbf{u}}(t)^T \mathbf{J}(t){\mathbf{u}}(t)}\right){\mathbf{u}}(t), \nonumber \\ 
& = \left( \mathbf{I} - \mathbf{u}(t)\mathbf{u}(t)^T \right)\mathbf{J}(t)\mathbf{u} \nonumber \\ 
& = \mathbb{P}(\mathbf{u}(t))\mathbf{J}(t)\mathbf{u}(t) \nonumber \\ 
&=\mathbf{f}_\mathbf{J}(\mathbf{u}(t)),
\end{align}

where
$ \mathbf{u}(t)^T \mathbf{J}(\mathbf{z}(t)), \mathbf{u}(t)
= \lambda(\mathbf{u}(t), \mathbf{z}(t)),
$ is the ICLE associated with the CLV $\mathbf{u}$ (Eq.~\ref{eq_6}), and
$\mathbb{P}({\mathbf{u}}) = \mathbf{I} - {\mathbf{u}(t)}{\mathbf{u}(t)}^T$ is an operator that projects any vector onto a direction orthogonal to ${\mathbf{u}(t)}$. Physically, $\dot{{\mathbf{u}}}$ is orthogonal to the state vector ${\mathbf{u}(t)}$, i.e., the solution of \eqref{eq_clv projeciton form} has constant norm, i.e., ${d}||{\mathbf{u}(t)}||/{dt} = 0$.
This means that the dynamics of CLVs lie on the surface of a unit sphere, $S^{d-1}$.

\section{Roadmap of the Theory}\label{sec:theory_roadmap}
A roadmap of the theoretical approach is shown in Figure \ref{fig:roadmap}.
By taking advantage of the geometric interpretation of the CLV equation \eqref{eq_clv projeciton form}, which defines a nonlinear system on the unit sphere (Section~\ref{section:clv_equation}), we identify and characterize its invariant sets. First, we assume Jacobians with stationary eigenbases, in which the fixed points of the CLV dynamics can be exactly linked to the eigenvectors of the Jacobian (Section~\ref{section: stability of clvs}). Second, we consider Jacobians with non-stationary eigenbases in fast-slow systems (Section \ref{section:fast-slow systems}), i.e., systems that have separate timescales. We argue that slowly varying eigendirections can be modelled as fixed points when the CLVs decay fast, which is referred to as the {\it adiabatic condition} (Section~\ref{sec:adiabatic_approx}). Finally, we explain the universal route to extreme events and critical transitions, which is revealed by the geometric structures of CLVs and eigenvectors, and their growth rates (Section~\ref{dynamical_regimes}).

\begin{figure}
\centering
\includegraphics[width=0.9\linewidth]{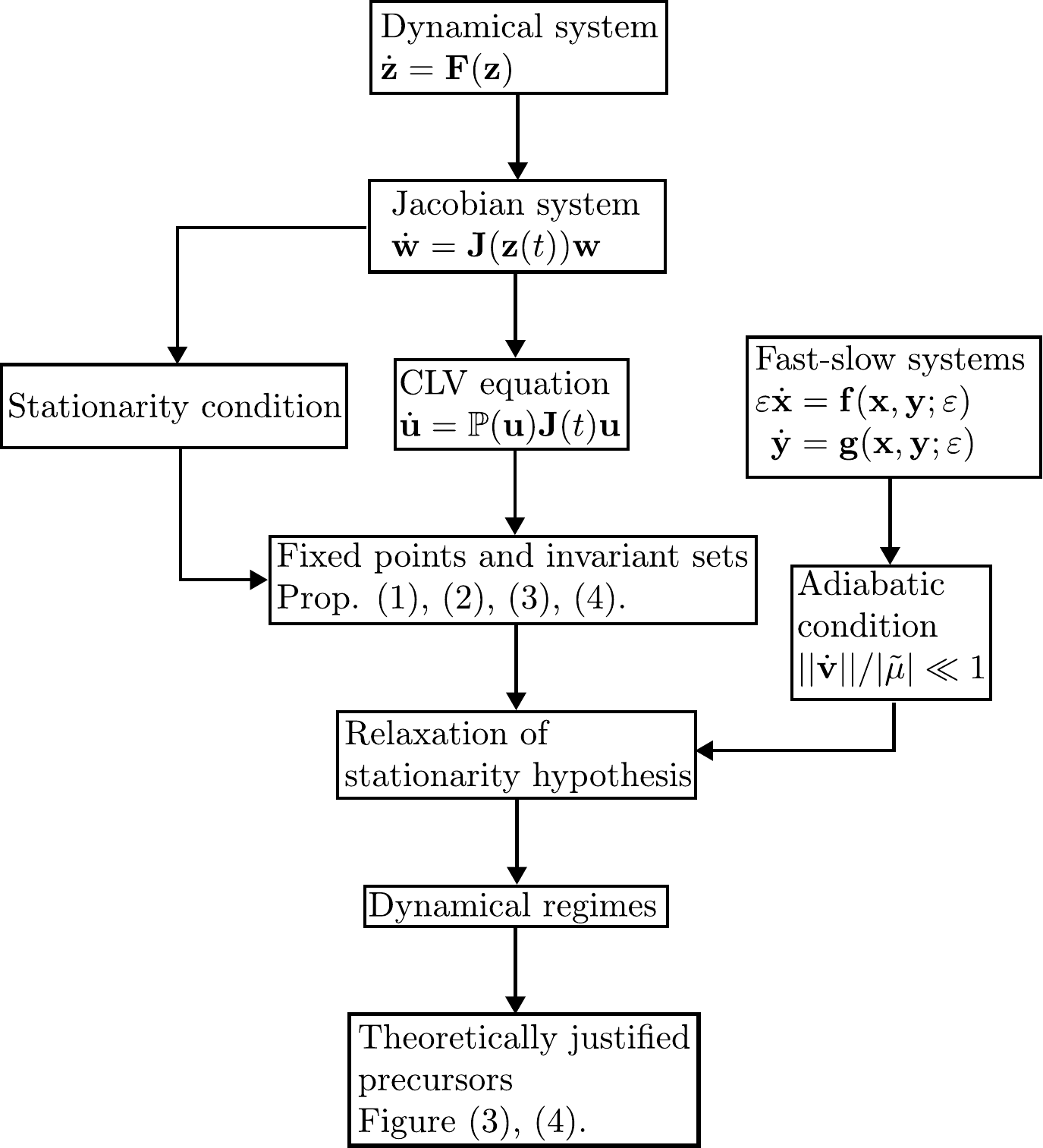}
\caption{Roadmap of the theoretical steps of Section \ref{section:clv_equation}, Section \ref{section:fast-slow systems}, and Section \ref{sec: adiabatic condition and dynamical regimes}.}
\label{fig:roadmap}
\end{figure}

\section{Properties and Invariant Sets of the CLV Dynamics}
\label{section: stability of clvs}

\subsection{Invariant Sets} \label{sec: fixed point and limit cycles}

We analyse the invariant sets of \eqref{eq_clv projeciton form} and their stability to gain insight into the behaviour of CLVs.
First, we analyse fixed points. Because Equation~\eqref{dynamical_system} contains only real quantities, the fixed points must also be real.
(For brevity, unless it is necessary for clarity, we omit the explicit dependence on time, i.e., $\mathbf{u}(t) = \mathbf{u}$.)

\begin{proposition} \label{prop: eig_fixed_points}
The points $\mathbf{v}$ and $-\mathbf{v}$ are fixed points of \eqref{eq_clv projeciton form} if and only if $\mathbf{v}$ is a real stationary eigenvector of $\mathbf{J}$ with unit magnitude.
\end{proposition}

\begin{proof}
($\Rightarrow$) Assume $\mathbf{v}$ is a fixed point of \eqref{eq_clv projeciton form}. Then $\mathbf{0} = \dot{\mathbf{v}} = \mathbf{Jv} - (\mathbf{v}^T\mathbf{Jv})\mathbf{v}$. Let $\mathbf{v}^T\mathbf{Jv} = \mu$ be a scalar value; then the expression becomes $\mathbf{Jv} = \mu \mathbf{v}$. Thus, $\mathbf{v}$ is an eigenvector of $\mathbf{J}$. Since $\mathbf{v}$ is a fixed point, we also have that $\dot{\mathbf{v}} = \mathbf{0}$, and therefore it is a stationary eigenvector of $\mathbf{J}$. The same argument can be used for $-\mathbf{v}$.
($\Leftarrow$) The proof is straightforward and left to the reader.
\end{proof}

If $\mathbf{J}$ has $k$ stationary eigenvectors with geometric multiplicity one, $g = 1$, then Equation \eqref{eq_clv projeciton form} has $2k$ fixed points, that is, two for each real stationary eigenvector, i.e., $\mathbf{v}_i$ and $-\mathbf{v}_i$ with $i=1,2,\ldots,k$ (left panel of Figure \ref{fig:CLV_sphere}).

\begin{proposition}
\label{proposition 2}
Consider a Jacobian $\mathbf{J}$ with an eigenvalue $\mu \in \mathbb{R}$ with geometric multiplicity $g(\mu)>1$. If the eigenspace $E$ associated with $\mu$ is stationary, the set $M = { \mathbf{u} \in E: ||\mathbf{u}|| = 1 } = E \bigcap S^{d-1}$ is a set of fixed points of Equation \eqref{eq_clv projeciton form}.
\end{proposition}

\begin{proof}
The eigenspace $E$ is spanned by the $g$ eigenvectors associated with $\mu$, $E = \operatorname{span} \{ \mathbf{v}_0, \dots, \mathbf{v}_{g-1} \}$.
Therefore, any linear combination of eigenvectors associated with $\mu$ is an eigenvector of $\mathbf{J}$, and therefore all elements of $E$ are eigenvectors.
From Proposition \ref{prop: eig_fixed_points}, all points in $M = E \bigcap S^{d-1}$ are fixed points of Eq. \eqref{eq_clv projeciton form}. Geometrically, the set $M = E \bigcap S^{d-1}$ is an equator of the sphere $S^{d-1}$ (middle panel of Figure \ref{fig:CLV_sphere}).
\end{proof}

We now analyse complex conjugate pairs of eigenvectors.

\begin{proposition}
\label{proposition 3}
Let $\mathbf{v}$ and ${\mathbf{v}^*}$ be two complex conjugate eigenvectors spanning a real stationary eigenspace $E$. Then the set $M_E = E \bigcap S^{d-1}$ is invariant under the flow associated with $\mathbf{f}_{\mathbf{J}}$ and generates a rotation.
\end{proposition}

\begin{proof}
Let $\sigma = \mu + i\omega$ and ${\sigma}^* = \mu - i\omega$ be the eigenvalues associated with $\mathbf{v}$ and ${\mathbf{v}^*}$, respectively. A real eigenspace $E$ is spanned by a linear combination $a\mathbf{q}+b\mathbf{p}$, where $\mathbf{q} = \mathbf{v} + \mathbf{v}^* = 2\mathrm{Re}(\mathbf{v})$ and $\mathbf{p} = i(\mathbf{v} - {\mathbf{v}}^*)=-2\mathrm{Im}(\mathbf{v})$, and $a,b$ are arbitrary real numbers.
From Equation \eqref{eq_clv projeciton form}

\begin{align}
&\mathbf{f_J}(a\mathbf{q} + b\mathbf{p}) = \nonumber \\ 
&=\mathbf{J} (a\mathbf{q} + b\mathbf{p}) +(a\mathbf{q}+ b\mathbf{p})^T\mathbf{J}(a\mathbf{q} + b\mathbf{p}) (a\mathbf{q} + b\mathbf{p})\nonumber \\ 
&= a(\mu \mathbf{q} + \omega\mathbf{p}) + b(\mu \mathbf{p} - \omega\mathbf{q}) + c(a\mathbf{q} + b\mathbf{p}) , \in E
\end{align}

where $a,b,c$ are real numbers. Since any linear combination of vectors in $E$ is mapped into it, $E$ is invariant under $\mathbf{f_J}$.
Consequently, because the set $M_E = S^{d-1} \bigcap E$ is invariant under the flow map of $\mathbf{f_J}$, and both the sphere $S^{d-1}$ and the eigenspace $E$ are invariant, their intersection must also be invariant.
Hence, $M_E = { \mathbf{u} \in E: ||\mathbf{u}|| = 1 }$ is a circle $S^1$ (an equator) along which the only admissible dynamics are rotations (right panel of Figure \ref{fig:CLV_sphere}).
\end{proof}

\subsection{Stability of CLV Dynamics and the Projected CLV Operator}
\label{subsection::linear_stability}

We now analyse the linearized dynamics of CLVs around an arbitrary fixed point, $\mathbf{v}_i$, from Propositions \ref{prop: eig_fixed_points} and \ref{proposition 2}.
The time evolution of an infinitesimal perturbation to the CLV, $\delta \mathbf{u}$, is governed by the linearized version of \eqref{eq_clv projeciton form}, which reads

\begin{align} \label{eq: full_linearized_system}
\frac{d \delta\mathbf{u}}{dt} &= \mathbf{J} \delta \mathbf{u} - \left( \mathbf{v}_i^T\mathbf{J}\delta \mathbf{u}\right)\mathbf{v}_i - \left(\mathbf{v}_i^T\mathbf{J}\mathbf{v}_i\right) \delta \mathbf{u}
\nonumber\\ &= \left[\left(\mathbf{I} - \mathbf{v}_i\mathbf{v}_i^T \right)\mathbf{J} - \mu_i\mathbf{I} \right]\delta \mathbf{u} \nonumber\\ &= \mathbf{L}_{\mathbf{v}_i}\delta \mathbf{u}
\end{align}

where $\mathbf{L}_{\mathbf{v}_i}$ is the \textit{projected CLV operator (PCLO)} associated with the eigenvector $\mathbf{v}_i$ and eigenvalue ${\sigma_i = \mu_i +i\omega_i}$. To derive \eqref{eq: full_linearized_system}, we exploited the fact that the perturbation $\delta \mathbf{u}$ is orthogonal to $\mathbf{u}$ because of the norm constraint imposed by Equation \eqref{eq_clv projeciton form}, i.e., the perturbation cannot increase the norm; therefore, $\delta \mathbf{u}^T\mathbf{u} = 0$.

We analyse the CLV stability by calculating the eigenvalues and eigenvectors of the PCLO \eqref{eq: full_linearized_system}.

\begin{proposition}
The eigenvectors associated with the PCLO, $\mathbf{L}_{\mathbf{v}_i}$, are $\tilde{\mathbf{v}}_i^{(j)} = \left( \mathbf{I} - \mathbf{v}_i\mathbf{v}_i^T\right)\mathbf{v}_{j}$, where the subscript $i$ indicates the $i$th PCLO and the superscript $(j)$ indicates the $j$th eigenvector of the PCLO. Therefore, the PCLO eigenvectors are the orthogonal projections of the eigenvectors $\mathbf{v}_j$ of $\mathbf{J}$.
The eigenvalues of $\mathbf{L}_{\mathbf{v}_i}$ are $\tilde{\sigma}_i^{(j)} = \sigma_j - \sigma_i$, where $\sigma_i$ and $\sigma_j$ are the eigenvalues of $\mathbf{J}$.
\end{proposition}

\begin{proof}
$\delta\mathbf{u} = \left( \mathbf{I} - \mathbf{v}_i\mathbf{v}_i^T\right)\delta\mathbf{u}$ because of orthogonality. We express the perturbation in the eigenbasis as $\delta \mathbf{u} = \mathbf{V}\mathbf{a}$, where $\mathbf{V} = \left[\mathbf{v}_0, \mathbf{v}_1, \ldots, \mathbf{v}_{d-1} \right]$ is the matrix of eigenvectors of $\mathbf{J}$, and $\mathbf{a}$ are the components in the eigenbasis.
Hence \footnote{The vector $\mathbf{J}\,\delta\mathbf{u}$ is real because both the Jacobian $\mathbf{J}$ and the perturbation $\delta\mathbf{u}$ are real. Writing $\delta\mathbf{u}$ in the Jacobian eigenbasis gives $\mathbf{J}\,\delta\mathbf{u} = \mathbf{V}\mathbf{\Lambda}\mathbf{a}$, where $\mathbf{V}$, $\mathbf{\Lambda}$, and $\mathbf{a}$ are generally complex. However, since $\mathbf{J}$ is real, these complex quantities occur in conjugate pairs, so the imaginary parts cancel and $\mathbf{J}\,\delta\mathbf{u}$ remains real.}
\begin{align} \label{eq: Derivation of eigenvectors}
& \mathbf{L}_{\mathbf{v}_i}\left( \mathbf{I} - \mathbf{v}_i\mathbf{v}_i^T\right)\delta \mathbf{u} =\nonumber \\ &= \mathbf{L}_{\mathbf{v}_i} \delta\mathbf{u} \nonumber\\ &= \left( \mathbf{I} - \mathbf{v}_i\mathbf{v}_i^T\right)\mathbf{J}\delta \mathbf{u} - \mu_i \delta\mathbf{u} \nonumber\\
&= \left( \mathbf{I} - \mathbf{v}_i\mathbf{v}_i^T\right)\mathbf{J}\mathbf{Va} - \left( \mathbf{I} - \mathbf{v}_i\mathbf{v}_i^T\right)\mu_i \mathbf{Va} \nonumber\\
&= \left( \mathbf{I} - \mathbf{v}_i\mathbf{v}_i^T\right) \mathbf{V} \left(\mathbf{\Lambda} - \mu_i \mathbf{I} \right) \mathbf{a}.
\end{align}
Let $\widetilde{\mathbf{V}}_{i} = \left( \mathbf{I} - \mathbf{v}_i\mathbf{v}_i^T\right) \mathbf{V}$ and $\widetilde{\mathbf{\Lambda}}_{i} = \mathbf{\Lambda} - \mu_i \mathbf{I}$. From \eqref{eq: Derivation of eigenvectors}, $\mathbf{L}_{\mathbf{v}_i}\widetilde{\mathbf{V}}_{i}\mathbf{a} = \widetilde{\mathbf{V}}_{i}\widetilde{\mathbf{\Lambda}}_{i}\mathbf{a}$, which implies that $\widetilde{\mathbf{V}}_{i}$ and $\widetilde{\mathbf{\Lambda}}_{i}$ are the matrices of eigenvectors and eigenvalues of $\mathbf{L}_{\mathbf{v}_i}$, respectively.
\end{proof}

The matrices $\widetilde{\mathbf{V}}_i$ and $\widetilde{\mathbf{\Lambda}}_i$ are singular because the projection operator applied to $\mathbf{v}_i$ returns a trivial vector

\begin{align}
\tilde{\mathbf{V}}_i &= \left(\mathbf{I} - \mathbf{v}_i\mathbf{v}_i^T\right)\mathbf{V} \nonumber\\
&= \left(\mathbf{I} - \mathbf{v}_i\mathbf{v}_i^T\right)[\mathbf{v}_0, ..., \mathbf{v}_i,...,\mathbf{v}_{d-1}] \nonumber\\
&= [\left(\mathbf{I} - \mathbf{v}_i\mathbf{v}_i^T\right)\mathbf{v}_0, ..., 0, ..., \left(\mathbf{I} - \mathbf{v}_i\mathbf{v}_i^T\right)\mathbf{v}_{d-1}],
\end{align}

with the corresponding eigenvalue being zero

\begin{align} \label{eigenvalue_PCLO}
\tilde{\mathbf{\Lambda}}_i &= \text{diag}\left(\sigma_0, ...,\sigma_i, ..., \sigma_{d-1} \right) - \mu_i\mathbf{I} \nonumber\\ &= \text{diag}\left(\sigma_0 - \mu_i, ..., 0, ..., \sigma_{d-1} - \mu_i \right).
\end{align}

This is expected, as the $d$-dimensional perturbation vector $\delta \mathbf{u} \in \mathbb{R}^d$ is constrained in the $(d-1)$-dimensional tangent space of the CLV sphere $T_{\mathbf{v}_i}S^{d-1}$, that is, the orthogonal space to the vector $\mathbf{v}_i$. Since $\mathbf{L}_{\mathbf{v}_i}$ is an operator that evolves the $d$-dimensional vector $\delta \mathbf{u}$ in a $(d-1)$-dimensional space, the operator must be singular.

The {real part of the spectrum} of the PCLO consists of $\tilde{\mu}_i^{(j)} = \mu_j - \mu_i$, with $i,j=0,1,\ldots,d-1$, which quantifies the {relative convergence or divergence rate of perturbations around $\mathbf{v}_i$}. Hence, a fixed point, $\mathbf{v}_i$, is stable if the {real parts of all the non-trivial eigenvalues} of the linearized operator $\mathbf{L}_{\mathbf{v}_i}$ are negative. There are two cases. First, if the geometric multiplicity is $g=1$, there can be only one stable fixed point $\mathbf{v}_i$ corresponding to the {eigenvalue with the largest real part $\mu_i$}; the entries of $\mathrm{Re}(\tilde{\mathbf{\Lambda}}_i) = \text{diag}(\mu_0 - \mu_i, ..., \mu_{d-1}-\mu_i)$ are all negative because $\mu_i > \mu_j$ for any $j$ from $0$ to $d-1$, except when $j=i$, in which case the entry is zero. This means that the eigenvector with the largest growth rate is the unique attracting direction (fixed point) of the CLVs on the unit sphere. Thus, all nearby CLVs are attracted by this eigenvector, which geometrically leads to CLV tangency along the dominant eigenvector (Proposition \ref{prop: eig_fixed_points}, left panel of Figure \ref{fig:CLV_sphere}).
Second, if the multiplicity is $g > 1$, the CLVs are attracted to the hyperplane (as opposed to a point) spanned by the $g > 1$ eigenvectors (Proposition \ref{proposition 2}, middle panel of Figure \ref{fig:CLV_sphere}).

In conclusion, the stability of a fixed point, $\mathbf{v}_i$, is dictated by the \emph{spectral gaps} in \eqref{eigenvalue_PCLO}. Large negative gaps correspond to rapid decay of CLVs toward $\mathbf{v}_i$, whereas small gaps correspond to weak attraction. Spectral gaps will play a central role in Section \ref{sec:adiabatic_approx}. Table \ref{tab:Jacobian and CLVs} summarizes the connections between the CLV dynamics and the eigenstructure of the Jacobian.

\begin{table}
\caption{Summary of the connection between the properties of the Jacobian matrix and CLV dynamics. }
\label{tab:Jacobian and CLVs}       
\begin{tabular}{| p{0.225\textwidth} | p{0.225\textwidth} |}
    \hline
    \multicolumn{1}{|c|}{\textbf{Jacobian} $\mathbf{J}(t)$} & \multicolumn{1}{c|}{\textbf{CLV dynamics} $\mathbf{u}(t)$} \\
    \hline
    Stationary eigenvector with geometric multiplicity $g= 1$ & Discrete fixed points \\
    \hline
    Stationary eigenvector with geometric multiplicity $g> 1$ & Higher dimensional invariant set\\
    \hline
    Complex conjugate eigenvalue pairs & Rotations around an equator \\
    \hline
    Spectral gaps & \textcolor{black}{Convergence rates}\\ 
    \hline
  \end{tabular}
\end{table}

%Consider also that the stationarity of $\mathbf{u}_0$ doesn't imply the stationarity of $\mu_0$, that is both $\mathbf{J}(t)$ and $\mu_0(t)$ may change in time, and the stability of the linear operator may change in time.

\begin{figure}
    \centering
    \includegraphics[width=1\linewidth]{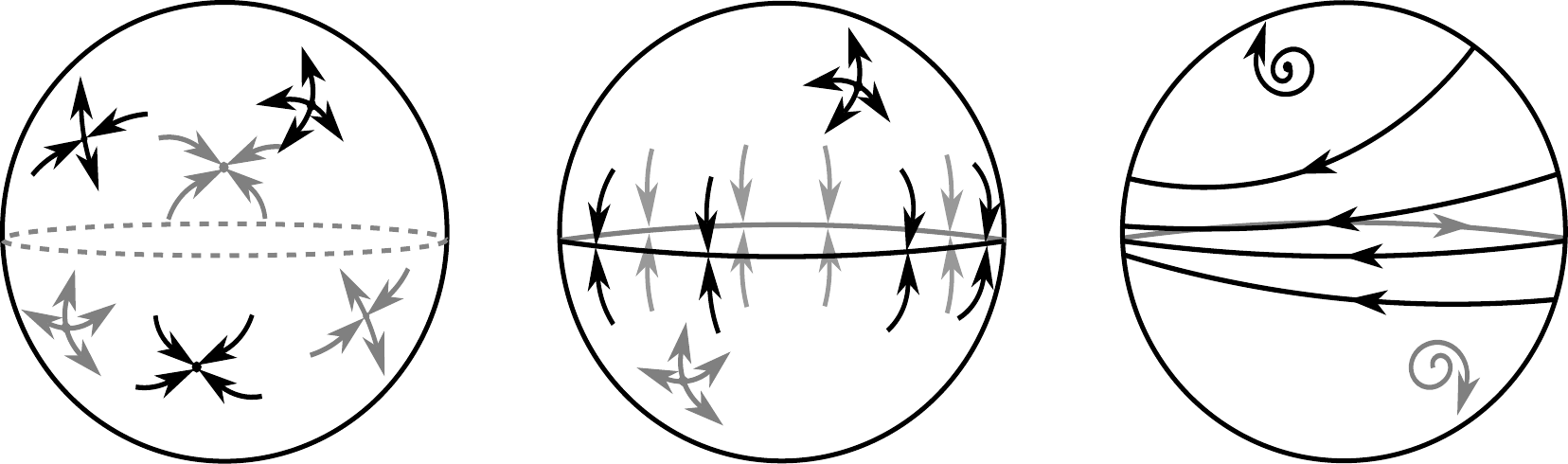}
    \caption{On the left: CLV phase space structure when the eigenvalues of the Jacobian $\mathbf{J}$ of Equation \eqref{dynamical_system} are real and follow the inequality $\mu_1 > \mu_2 > \mu_3$. In the middle: CLV phase space structure when the eigenvalues are real and follow $\mu_1 = \mu_2>\mu_3$. On the right: CLV phase space structure when $\sigma_1 = \sigma_2^*$ are complex conjugate and $\mathrm{Re}(\sigma_1) = \mathrm{Re}(\sigma_2) > \mu_3$. In black, we indicate the fixed points and invariant sets on the foreground hemisphere, in gray we indicate fixed points and invariant sets in the background hemisphere.}
    \label{fig:CLV_sphere}
\end{figure}

\section{Fast-Slow Systems} \label{section:fast-slow systems}
Fast-slow systems are dynamical systems with multiple time scales, which can be described by equations in which some variables vary more rapidly than others \cite{kuehnMathematicalFrameworkCritical2011, kuehnMultipleTimeScale2015}. {An extreme event or a critical transition can occur when the slow manifold becomes locally unstable (Figure \ref{fig: manifold}).} In fast-slow systems, the state vector of the original system \eqref{dynamical_system} can be partitioned as $\mathbf{z}=[\mathbf{x},\mathbf{y}]$ with governing equations

\begin{equation} \label{fast-slow-t}
    \begin{aligned}
     \varepsilon   \dot{\mathbf{x}} &= \mathbf{f}(\mathbf{x}, \mathbf{y}; \varepsilon) \\
      \dot{\mathbf{y}} &= \mathbf{g}(\mathbf{x}, \mathbf{y}; \varepsilon),
    \end{aligned}
\end{equation}
where $\dot{}=\frac{d}{d\tau}$ is the time derivative with respect to the slow time $\tau$, $0 < \varepsilon\ll1$ is the \textit{time-scale parameter}, $\mathbf{x} \in \mathbb{R}^m$ are the \textit{fast variables}, and $\mathbf{y} \in \mathbb{R}^n$ are the \textit{slow variables}, with $n + m = d$. The functions $\mathbf{f}: \mathbb{R}^m\times\mathbb{R}^n \to \mathbb{R}^m$ and $\mathbf{g}: \mathbb{R}^m\times\mathbb{R}^n \to \mathbb{R}^n$ are at least as smooth as $\mathbf{F}$ in Equation \eqref{dynamical_system}. 
By performing a change of time variable $t = \tau/\varepsilon$, \eqref{fast-slow-t} is recast as a standard perturbation problem \cite{kuehnMultipleTimeScale2015}

\begin{equation} \label{fast-slow-tau}
    \begin{aligned}
      \mathbf{x}' &= \mathbf{f}(\mathbf{x}, \mathbf{y}; \varepsilon) \\
      \mathbf{y}' &= \varepsilon \mathbf{g}(\mathbf{x}, \mathbf{y}; \varepsilon),
    \end{aligned}
\end{equation}
where $'$ indicates the time derivative with respect to the fast time scale $t$. 
{In the singular limit $\varepsilon = 0$, \eqref{fast-slow-t} becomes the \textit{slow subsystem}} \cite{kuehnMultipleTimeScale2015}

\begin{equation} \label{slow-subsystem}
    \begin{aligned}
        \mathbf{0} &= \mathbf{f}(\mathbf{x},\mathbf{y};0), \\
        \dot{\mathbf{y}} &= \mathbf{g}(\mathbf{x},\mathbf{y};0), 
    \end{aligned}
\end{equation}

which is a set of differential-algebraic equations, where $\mathbf{f}(\mathbf{x},\mathbf{y};0) = \mathbf{0}$ defines a set on which the slow dynamics take place. The set $C_0 = \{(\mathbf{x},\mathbf{y}) \in \mathbb{R}^m \times \mathbb{R}^n:\mathbf{f}(\mathbf{x},\mathbf{y};0) = \mathbf{0}\}$ is referred to as the \textit{critical manifold}.
The solution of the slow subsystem is the \textit{slow flow}. On the other hand, when $\varepsilon=0$ in the fast-time formulation \eqref{fast-slow-t}, we obtain the \textit{fast subsystem}

\begin{equation} \label{fast-subsystem}
    \begin{aligned}
        \mathbf{x}' &= \mathbf{f}(\mathbf{x},\mathbf{y};0),\\
        \mathbf{y}' &= \mathbf{0}.
    \end{aligned}
\end{equation}

The solution of the fast subsystem is the \textit{fast flow}, in which the set of fixed points is the critical set $(\mathbf{x}_0,\mathbf{y}_0) \in C_0$. 
We consider normally hyperbolic invariant manifolds, which are subsets $S_0 \subset C_0$ over which the Jacobian of the fast vector field with respect to the fast variables, $\nabla_\mathbf{x}\mathbf{f}$, has no eigenvalues with zero real part. Such manifolds are hyperbolic fixed points of $\mathbf{x}' = \mathbf{f}(\mathbf{x},\mathbf{y};0)$. Normally hyperbolic invariant manifolds are attracting if the eigenvalues in the spectrum of $\nabla_\mathbf{x}\mathbf{f}(\mathbf{x}^*,\mathbf{y}^*)$, with $(\mathbf{x}^*,\mathbf{y}^*)\in S_0$, have negative real parts, whereas they are repelling if they {have positive real parts}. 
Fenichel's theorem \cite{fenichelGeometricSingularPerturbation1979, kuehnMathematicalFrameworkCritical2011, kuehnMultipleTimeScale2015} guarantees the existence of a locally invariant manifold $S_{\varepsilon}$ diffeomorphic to $S_0$ that is $\mathcal{O}(\varepsilon)$ close to $S_0$ as the time-scale parameter $\varepsilon \to 0$. Such a manifold is the \textit{slow manifold}. Fenichel's theorem ensures that solutions obtained by analysing the fast and slow subsystems independently in the singular limit can be combined to yield a solution that is 
$\mathcal{O}(\varepsilon)$-close to {the solution of Equation \eqref{fast-slow-t}}.

\subsection{Tangent Space} \label{subsection: tangent space}

We analyse the tangent dynamics of fast-slow systems. This decomposition allows us to apply Fenichel's theorem and describe the behaviour of CLVs. We consider the tangent dynamics associated with the slow subsystem \eqref{slow-subsystem}

\begin{subequations} \label{slow-tangent_dynamics}
\begin{align}
    \delta \mathbf{x} &= -(\nabla_{\mathbf{x}}\mathbf{f})^{-1}\nabla_{\mathbf{y}}\mathbf{f}\,\delta\mathbf{y}, \label{slow-tangent_dynamics:a} \\
    \delta\dot{\mathbf{y}} &= \left(-\nabla_{\mathbf{x}}\mathbf{g}(\nabla_{\mathbf{x}}\mathbf{f})^{-1}\nabla_{\mathbf{y}}\mathbf{f}
    + \nabla_\mathbf{y}\mathbf{g}\right)\delta\mathbf{y}. \label{slow-tangent_dynamics:b}
\end{align}
\end{subequations}

Equations \eqref{slow-tangent_dynamics} are the Jacobian equations of the reduced slow dynamics \eqref{slow-subsystem}, which describe the evolution of infinitesimal perturbations within the $n$-dimensional tangent space of the critical manifold, $\mathbf{f}(\mathbf{x},\mathbf{y};0)=\mathbf{0}$. This space is the \emph{slow tangent space}. Let us consider a generic perturbation $\delta \mathbf{z} = \left(\delta \mathbf{x}, \delta \mathbf{y} \right)^T \in \mathbb{R}^d$ constrained to the slow tangent space. Equation \eqref{slow-tangent_dynamics:a} shows that perturbations in the fast directions, {$\delta \mathbf{x}$}, are {uniquely determined by perturbations in the slow direction, $\delta \mathbf{y}$,} through a linear relation. {Consequently, the {tangent space of the slow manifold} is spanned by vectors of the form
\begin{equation} \label{slow_eogenspace}
\mathbf{v}_{s}
=
\left(
- (\nabla_\mathbf{x}\mathbf{f})^{-1}\nabla_\mathbf{y}\mathbf{f}\,\mathbf{v}_{\nu},
\;
\mathbf{v}_{\nu}
\right)^{\!T},
\end{equation}
{where {$\mathbf{v}_{\nu} \in \mathbb{R}^n$} is an arbitrary vector in the direction of the slow variables and $\mathbf{v}_s \in \mathbb{R}^{n+m}$ is obtained by lifting a slow perturbation $\mathbf{v}_{\nu}$ to the full space using the constraint imposed by the manifold.}}

Using the framework developed in Section \ref{section:clv_equation}, Equation \eqref{slow-tangent_dynamics:b} supports the existence of $n$ slow CLVs, which evolve on an $(n-1)$-dimensional sphere embedded in the slow tangent space, together with their associated slow Lyapunov exponents. Slow CLVs are elements of the $(n-1)$-dimensional sphere embedded in a $d$-dimensional space, with $d = m+n$; their representation in $d$ dimensions is $\mathbf{u}_{s}= \left(-(\nabla_\mathbf{x}\mathbf{f})^{-1}\nabla_\mathbf{y}\mathbf{f}\,\mathbf{u}_{\nu},\; \mathbf{u}_{\nu} \right)^{\!T}$, where $\mathbf{u}_{\nu}$ is the $n$-dimensional CLV. Beyond this general geometric characterisation, further analytical insight into the structure of the slow CLVs requires additional assumptions on the functional form of $\mathbf{f}$ and $\mathbf{g}$ and on the properties of their Jacobians.
To complete the dynamical picture of CLVs in fast-slow systems, we consider the tangent space dynamics of the fast subsystem \eqref{fast-subsystem}

\begin{equation} \label{fast-tangent_dynamics}
    \begin{aligned}
        \delta \dot{\mathbf{x}} &= \nabla_{\mathbf{x}}\mathbf{f} \,\delta \mathbf{x} + \nabla_{\mathbf{y}}\mathbf{f}\, \delta \mathbf{y} \\
        \delta\dot{\mathbf{y}} &= \mathbf{0},
    \end{aligned}
\end{equation}
with the Jacobian being 

\begin{align} \label{fast-slow jacobian}
\mathbf{J}_f =
\begin{bmatrix}
\nabla_\mathbf{x}\mathbf{f} & \nabla_\mathbf{y}\mathbf{f} \\
0 & 0
\end{bmatrix}.
\end{align}

The form of the Jacobian shows the existence of a splitting of the tangent space of the phase space on the slow manifold.  
$\mathbf{J}_f$ is a block triangular matrix and its non-trivial eigenvalues are the eigenvalues of $\nabla_\mathbf{x}\mathbf{f}$, which are non-zero under normal hyperbolicity assumptions. We refer to the first $m$ eigenvalues as the \textit{fast eigenvalues} $\eta_1, ..., \eta_m$ and to the trivial eigenvalues $\nu_{1},...,\nu_{n}=0$ as the \textit{slow eigenvalues}. If the slow manifold is attracting, the fast eigenvalues are all negative, $\eta_1,...,\eta_m < 0$. 

The eigenspaces associated with these eigenvalues are the \textit{fast eigenspace} $E_f$ and \textit{slow eigenspace} $E_s$, respectively. Under normal hyperbolicity assumptions, the matrix $\nabla_\mathbf{x}\mathbf{f}$ is non-singular; therefore, the fast and slow eigenspaces must be linearly independent of each other, that is, the hyperplanes defined by $E_f$ and $E_s$ must be transversal to each other, and their principal angle is non-zero.

The fast eigenvectors, i.e., the eigenvectors associated with the fast eigenvalues, can be computed from the Jacobian \eqref{fast-slow jacobian} and have a simple form $\mathbf{v}_{f} = \left( \mathbf{v}_{\eta}, \mathbf{0}\right)^T$, where $\mathbf{v}_{\eta}$ denotes the eigenvectors of the Jacobian $\nabla_\mathbf{x} \mathbf{f}$
\begin{equation} 
\begin{bmatrix}
\nabla_\mathbf{x}\mathbf{f} & \nabla_\mathbf{y}\mathbf{f} \\
0 & 0
\end{bmatrix}  
\begin{bmatrix}
    \mathbf{v}_{\eta} \\
    \mathbf{0}
\end{bmatrix} = 
\begin{bmatrix}
    \nabla_\mathbf{x}\mathbf{f} \, \mathbf{v}_{\eta} \\
    \mathbf{0}
\end{bmatrix} = \eta \mathbf{v}_{f}\, .
\end{equation} \label{eq: fast_eigenvec}

If $\nabla_\mathbf{x} \mathbf{f}$ is non-singular, the eigenvectors $\mathbf{v}_{\eta}$ form a basis of $\mathbb{R}^m$, and therefore $\mathbf{v}_{f}$ spans the eigenspace $\mathbb{R}^m$. Hence, if the hypothesis of normal hyperbolicity is satisfied, the fast eigenspace is stationary. For the slow eigenvectors, we find that $\mathbf{w}_{s} = \left(- (\nabla_\mathbf{x}\mathbf{f})^{-1}\nabla_\mathbf{y}\mathbf{f} \,\mathbf{w}_{\nu}, \mathbf{w}_{\nu} \right)^T$, where $\mathbf{w}_{\nu}$ can be any vector in $\mathbb{R}^n$,

\begin{align} \label{fast kernel}
\begin{aligned}
&\begin{bmatrix}
\nabla_\mathbf{x}\mathbf{f} & \nabla_\mathbf{y}\mathbf{f} \\
0 & 0
\end{bmatrix}  
\begin{bmatrix}
    - (\nabla_\mathbf{x}\mathbf{f})^{-1}\nabla_\mathbf{y}\mathbf{f} \,\mathbf{w}_{\nu} \\
    \mathbf{w}_{\nu}
\end{bmatrix} 
= \\
&=\begin{bmatrix}
    -\nabla_\mathbf{y}\mathbf{f} \, \mathbf{w}_{\nu} + \nabla_\mathbf{y}\mathbf{f} \, \mathbf{w}_{\nu} \\
    \mathbf{0}
\end{bmatrix} = \mathbf{0} \, , \; \forall \, \mathbf{w}_{\nu} \in \mathbb{R}^n \, .
\end{aligned}
\end{align}

The formula for $\mathbf{w}_{s}$ is well defined under the normal hyperbolicity hypothesis, as $\nabla_\mathbf{x}\mathbf{f}$ is non-singular. 
The family of vectors defined by $\mathbf{w}_{s}$ spans the kernel of the Jacobian in the fast-time formulation $\mathbf{J}_f$, as shown in Equation \eqref{fast kernel}, and it coincides with the slow eigenspace defined in \eqref{slow_eogenspace}, which corresponds to the tangent space of the slow manifold. Hence, the kernel of the Jacobian in fast time coincides with the tangent space of the slow manifold, $E_s = \text{Ker}(\mathbf{J}_f) = T_\mathbf{p}C_0$. 

We have therefore defined a decomposition of the tangent space of the phase space restricted to $C_0$. The decomposition highlights two distinct, well-defined spaces: the fast eigenspace $E_f$ and the slow eigenspace $E_s$. These two subspaces are transversal to each other and, therefore, their principal angle $\theta$ is non-zero as long as the state of \eqref{dynamical_system} is on the slow manifold. Such a decomposition is inherited from the fast-slow decomposition of the dynamical system. A schematic representing the splitting of the slow and fast subspaces is shown in Figure \ref{fig: manifold}. %\textcolor{red}{[I would hope we can make this entire section easier to digest. I continue finding this difficult to understand, especially what is the goal of this? What are all these new vectors being introduced? I am reading and thinking with the reviewer's hat.]}

\section{Universal Route to Extreme Events and Critical Transitions} \label{sec: adiabatic condition and dynamical regimes}

\subsection{Adiabatic Condition} 
\label{sec:adiabatic_approx}

As explained in Section \ref{subsection::linear_stability}, tangency between CLVs can be interpreted as the convergence of multiple CLV trajectories, which are solutions of Equation \eqref{eq_clv projeciton form}, toward a common stable fixed point (i.e., an eigenvector). To define fixed points and more general invariant sets, we showed that a stationarity hypothesis is necessary, that is, we required that eigenvectors, or spaces spanned by them, were stationary. However, this assumption breaks down if the underlying system \eqref{dynamical_system} has a non-stationary eigenbasis, which is the subject of this section.

Consider an eigenvector \( \mathbf{v}_i \) of \( \mathbf{J}(t) \) for which the eigenvalue matrix of the associated PCLO \eqref{eigenvalue_PCLO} has only non-positive entries, \( \tilde{\mu}_i^{(j)} \leq 0 \), so that \( \mathbf{v}_i \) is a stable zero of \( \mathbf{f}_{\mathbf{J}} \) from Equation \eqref{eq_clv projeciton form}. {If the rate of convergence} of CLVs toward this stable zero is much larger than the rate at which the eigenvector itself changes in time, then the eigenvector behaves approximately as a fixed point.
    Physically, this means that the eigenvector evolves slowly compared to the {convergence rate} of the CLVs: the CLVs decay toward the eigenvector before it has time to move significantly. As a result, the eigenvector acts approximately like a fixed point. We refer to this condition as the \textit{adiabatic condition}
\begin{equation} \label{adiabatic condition}
    \frac{\lVert \dot{\mathbf{v}} \rVert}{|\tilde{\mu}|} \ll 1,
\end{equation}
where \( \lVert \dot{\mathbf{v}} \rVert \) is the magnitude of the time derivative of the eigenvector \( \mathbf{v} \), and \( |\tilde{\mu}| \) is the smallest (in magnitude) eigenvalue of the PCLO associated with \( \mathbf{v} \) \eqref{eigenvalue_PCLO}. {The presence of multiple time scales in fast–slow systems leads to spectral gaps, as shown in Section \ref{subsection: tangent space}, which fulfil the adiabatic condition by construction.}  

\subsection{Definitions of Extreme Events and Critical Transitions} \label{sec: definition of extreme event}
Extreme events in fast-slow dynamical systems occur because the slow-dimensional manifold \cite{farazmandExtremeEventsMechanisms2019} becomes locally unstable, which means that at least one fast {eigenvalue becomes unstable}. Critical transitions occur when the slow dynamics bifurcate to another stable slow branch \cite{schefferEarlywarningSignalsCritical2009, schefferAnticipatingCriticalTransitions2012}. {Therefore, critical transitions are a subset of extreme events.}

\subsection{Route to Extreme Events and Critical Transitions} \label{dynamical_regimes}

{An extreme event or critical transition occurs when a fast eigenvalue changes stability (Section \ref{sec: definition of extreme event}). The stability change must pass through an intermediate near-neutral state, where the eigenvalue approaches zero. The onset of the extreme event is therefore expected to be organised by a sequence of dynamical changes. This motivates the existence of a \textit{sequence} or \textit{cascade} of \textit{regimes}: a \textit{slow regime}, a \textit{transition regime}, and a \textit{critical regime}.} 

\begin{enumerate}
\item {\textbf{Slow regime}}. When the state of the system~\eqref{fast-slow-t} evolves on a slow manifold, the fast variables are {driven by} the slow variables. The slow manifold~\eqref{fast-subsystem}, therefore, {is} approximately the set of fixed points of the fast dynamics (Section \ref{section:fast-slow systems}). At the fixed points, the CLVs coincide with the eigenvectors of the Jacobian~\eqref{fast-slow jacobian}. As a result, CLVs associated with fast time scales become tangent to the eigenvectors of the fast subsystem. In contrast, CLVs associated with the slow dynamics evolve according to the linearisation of the reduced slow flow (Equation~\eqref{slow-tangent_dynamics}), and lie in the kernel of the fast-time Jacobian (Section~\ref{subsection: tangent space}), which is the tangent space of the slow manifold (Section~\ref{subsection: tangent space}). Under the assumption that the critical manifold is normally hyperbolic and invariant, the fast and slow eigenspaces remain transversal (Section~\ref{subsection: tangent space}) {to each other}, preventing the formation of tangencies both among fast CLVs and between fast and slow CLV subspaces while the system remains on the slow manifold.

\item {\textbf{Transition regime}}. {Before a hyperbolic singularity, one or more fast eigenvalues of $\mathbf{J}$, which are initially negative, approach zero. In the short period of transition from negative to positive eigenvalues, one or more CLVs undergo a \textit{transition regime} and still follow Equation~\eqref{eq_clv projeciton form}. {During the transition regime, a CLV that is initially associated with the fast dynamics undergoes a change in time scale as the corresponding fast eigenvalue approaches zero. As the separation between fast and slow eigenvalues is lost, the CLV separates from the fast eigenvectors of the Jacobian.} When this occurs, two distinct behaviours may arise, both governed by the CLV evolution Equation~\eqref{eq_clv projeciton form} and the spectral properties of the Jacobian. In the first case, the fast CLV is repelled from the unstable fast eigenspace and collapses toward the tangent space of the slow manifold, which coincides with the kernel of the fast-time Jacobian (Section~\ref{subsection: tangent space}), in accordance with Proposition \ref{proposition 2}. This collapse produces a tangency between the fast and slow CLV subspaces. In the second case, the interaction between a fast eigenvalue and a slow eigenvalue leads to the formation of a complex conjugate pair, which induces a rotational motion of the CLV on the unit sphere, as described in Proposition~\ref{proposition 3}. Both mechanisms signal a breakdown of the slow–fast separation, thereby serving as precursors of an impending extreme event or critical transition.}

\item {\textbf{Critical regime}}. After the hyperbolic singularity, at least one of the fast eigenvalues becomes positive and the normally hyperbolic invariant manifold on which the dynamics evolve becomes repelling and pushes the state away, generating an extreme event, as in Figure \ref{fig: manifold}. {After the transition regime, the dynamics enter a \textit{critical regime}, in which one or more fast eigenvalues become positive, $\mathrm{Re}(\eta) \gg \varepsilon$, and the state of Equation~\eqref{fast-slow-t} is repelled from the invariant manifold. When the adiabatic condition~\eqref{adiabatic condition}, $\frac{\|\dot{\mathbf{v}}\|}{|\tilde{\mu}|} \ll 1,$ is satisfied, the eigenvector associated with the most positive eigenvalue acts approximately as an attracting fixed point of the CLV dynamics, as discussed in Section~\ref{subsection::linear_stability}. As a result, multiple CLVs are drawn toward this direction, producing tangencies along the fast eigendirection. Consequently, the fast and slow CLV subspaces, which are transversal during the slow regime, can develop one or more vanishing principal angles during the extreme event, giving rise to tangencies as the system enters the critical regime.}

\end{enumerate} 
{These three regimes prior to an extreme event are {clearly separated} in explicit fast-slow systems because the spectral gap is large.}
%The sharp division in time scales generates large eigenvalue gaps, which in turn makes the adiabatic condition satisfied. However, many systems do not admit an explicit fast–slow decomposition, and in such cases there is no guarantee that the dynamical regimes described above will occur in a well-defined manner. Nevertheless, a wide class of systems exhibiting extreme or bursting events display an implicit separation of time scales, even in the absence of a clearly identifiable small parameter. In these systems, the same dynamical regimes can still emerge, albeit in a less sharply separated form than in explicitly fast–slow systems.}
We introduce pseudo-algorithms to compute two precursors: Precursor 1 in Figure \ref{angle_algo} and Precursor 2 in Figure \ref{lambda_algo}, based on the dynamical properties of CLVs in fast-slow systems.
\begin{figure}
    \centering
    \includegraphics[width=0.95\linewidth]{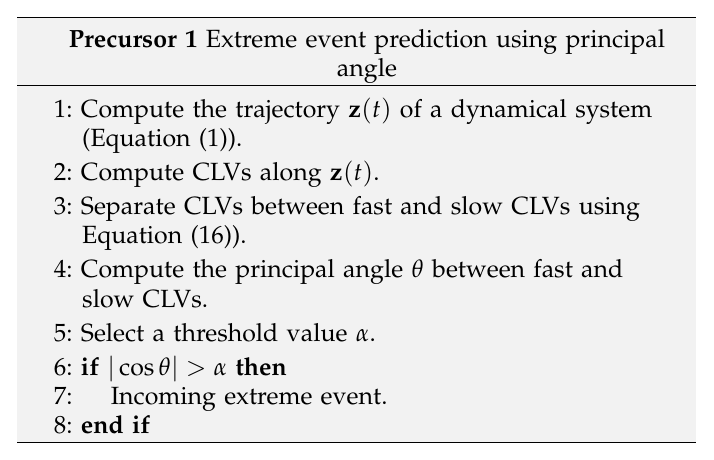}
    \caption{Schematic of the pseudo-algorithm for angle-based precursor.}
    \label{angle_algo}
\end{figure}
Finally, the cascade of events that we identified also explains the critical slowing down identified in  \cite{dakosSlowingEarlyWarning2008}. 

\begin{figure}
    \centering
    \includegraphics[width=0.95\linewidth]{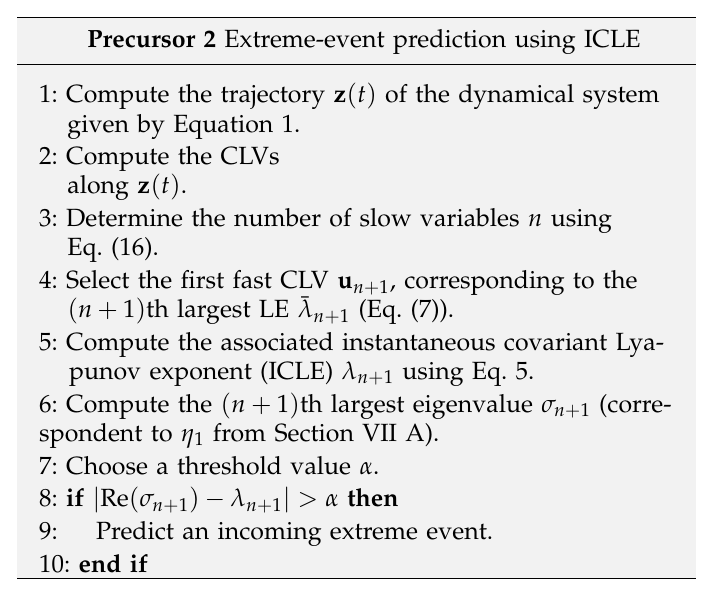}
    \caption{Schematic of pseudo-algorithm of ICLE-based precursor.}
    \label{lambda_algo}
\end{figure}

\begin{figure}
    \centering
    \includegraphics[width=0.95\linewidth]{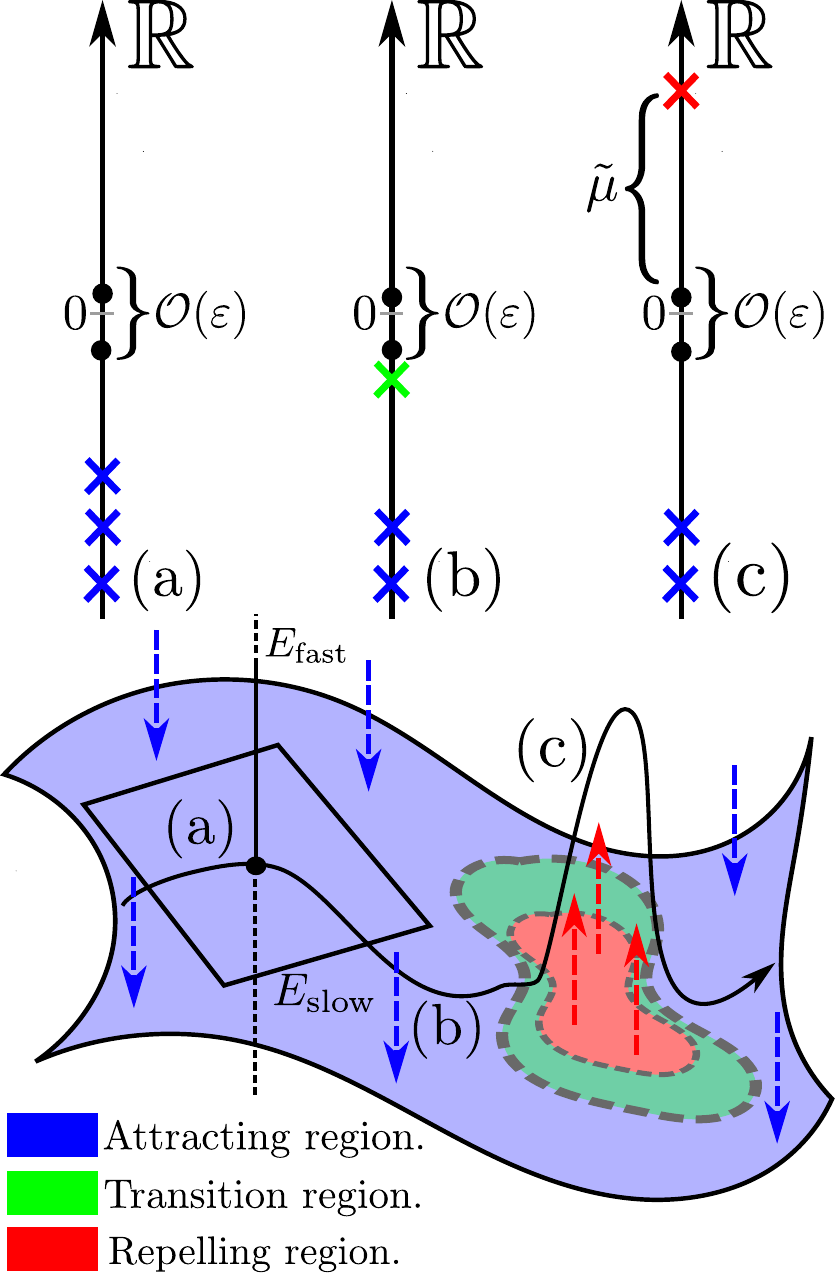}
    \caption{Top: real part of the spectrum of the Jacobian in the three dynamical regimes: the slow regime, the transition regime, and the critical regime. In blue, green, and red the fast eigenvalues, and in black the slow eigenvalues. {If the system is chaotic, at least one slow eigenvalue is positive on average}. Bottom: representation of the manifold in the three dynamical regimes. (a): slow regime, all fast eigenvalues are negative, the state is on the attracting critical manifold, and the fast CLVs are tangent to the fast eigenvectors. Fast and slow subspaces are transversal. (b): transition regime, one of the fast eigenvalues is about to become unstable and generates an instability. The critical manifold is about to become repelling. In this regime, one of the fast CLVs interacts with the slow CLVs, often via a tangency or a complex conjugation. This regime can be used as a precursor to the extreme event. (c): critical regime. The state is being repelled from the critical manifold. One of the fast eigenvalues becomes unstable and much larger than $\varepsilon$, making the corresponding fast eigenvector attracting. The strong attraction of the fast eigenvector generates tangencies between the fast and slow CLV spanned subspaces.}
    \label{fig: manifold}
\end{figure}

\subsection{An Analytical Example} \label{subsection: vanderpol}

The Van der Pol oscillator in its fast-slow form is a nonlinear dynamical system that exhibits critical transitions \cite{kuehnMathematicalFrameworkCritical2011, kuehnMultipleTimeScale2015}

\begin{equation} \label{slow_vdp}
\begin{aligned} 
\varepsilon \frac{dx}{d\tau} &= y - \frac{x^3}{3} + x, \\
\frac{dy}{d\tau} &= -x.
\end{aligned}
\end{equation}

The state of the system converges to a limit cycle.
The cubic term in the first equation introduces a dynamical fold bifurcation, which makes the state rapidly transition between two attracting branches of the critical manifold defined by the cubic $y = \frac{x^3}{3} - x$. 

The dynamics on the limit cycle can be decomposed into two subsystems (Section \ref{section:fast-slow systems}). The slow subsystem describes the dynamics on the stable slow {manifold}, represented by the attracting branches of the cubic. The fast subsystem describes the dynamics of the transitions between the two cubic branches. 
The eigenvalues of the Jacobian of \eqref{slow_vdp} in fast-time form are $\sigma_{1,2} = \left( \xi \mp \sqrt{\xi^2 - 4\varepsilon} \right)/2$ and the corresponding eigenvectors are $\mathbf{v}_{1,2} = \left( -\sigma_{1,2}, \varepsilon \right)^T$, where $\xi = 1-x^2$. 
We expand both eigenvalues in a Taylor series with respect to $\varepsilon$ around zero and find $\sigma_1 =\xi + \mathcal{O}(\varepsilon)$, $\sigma_2 = \mathcal{O}(\varepsilon)$ and for the eigenvectors $\mathbf{v}_1 = \left(1, 0 \right)^T + \mathcal{O}(\varepsilon)$, $\mathbf{v}_2 = \left(-1/\xi, 1 \right)^T + \mathcal{O}(\varepsilon)$. These computations are consistent with the eigenvalues and eigenvectors of the Jacobian of the fast subsystem in the critical limit. {We refer to the two CLVs as the stable CLV $\mathbf{u}_{st}$ and the neutral CLV $\mathbf{u}_n$, and to their corresponding LEs as the stable LE $\bar{\lambda}_{st} < 0$ and the neutral LE $\bar{\lambda}_n = 0$, respectively. This system is not chaotic and therefore does not have a positive LE.}

We identify the three regimes introduced in Section \ref{dynamical_regimes} for this particular system in Figure \ref{icle_eig_vdp}. 
\begin{enumerate}
\item The slow regime represents the dynamics on the slow manifold, when $\xi < 0$. As described in Section \ref{dynamical_regimes}, in this regime, the stable CLV $\mathbf{u}_{st}$ is tangent to the fast eigenvector, therefore $\mathbf{u}_{st} \approx \left(1,0\right)^T \approx \mathbf{v}_1$. Consequently, the stable ICLE is approximately equal to the real part of the first eigenvalue, $\lambda_{st} \approx \mu_1$. The neutral CLV is always tangent to the trajectory and, in this case, is also approximately tangent to the slow eigenvector $\mathbf{u}_n = \left( -1/\xi, 1\right)/\gamma(\xi) \approx \mathbf{v}_2$, where $\gamma(\xi)$ is a normalisation factor.
\item During the critical transition, the CLV dynamics enter the critical regime, where the real part of the fast eigenvalue is positive, $\mu_1 \gg \varepsilon$. The fast eigenvector remains approximately stationary during the critical transition, making it a stable fixed point of the CLV dynamics. The adiabatic condition is satisfied; thus, the neutral and stable CLVs become tangent to the dominant eigenvector $\mathbf{v}_f \approx\mathbf{u}_{st} \approx\mathbf{u}_n$. Thus, we have $\lambda_{st} \approx \lambda_n \approx \mu_1$.

\item {In this case, the transition regime is characterised by the interval $-2\sqrt{\varepsilon} \leq \xi \leq 2\sqrt{\varepsilon}$, where the two eigenvalues become complex conjugates. This transition regime cannot be captured by the critical limit, because the assumption of a clear separation of time scales between the two variables breaks down. The Jacobian spectrum develops complex conjugate eigenvalues and, in the CLV formulation, this spectral change induces rotational dynamics. The onset of this rotational behaviour reflects a coupling between fast and slow directions, which can be interpreted as an early indicator of the impending critical transition.} 
\end{enumerate}

Our predictions are tested with a simulation of system \eqref{slow_vdp}. We use $\varepsilon = 0.001$ and integrate the system using a fourth-order Runge–Kutta time integrator. 
Figure \ref{icle_eig_vdp} shows that, away from the critical transition, the ICLEs correspond to the eigenvalues of the Jacobian, which means that the CLVs are tangent to the eigenvectors.

\begin{figure}
    \centering
    \includegraphics[width=0.95\linewidth]{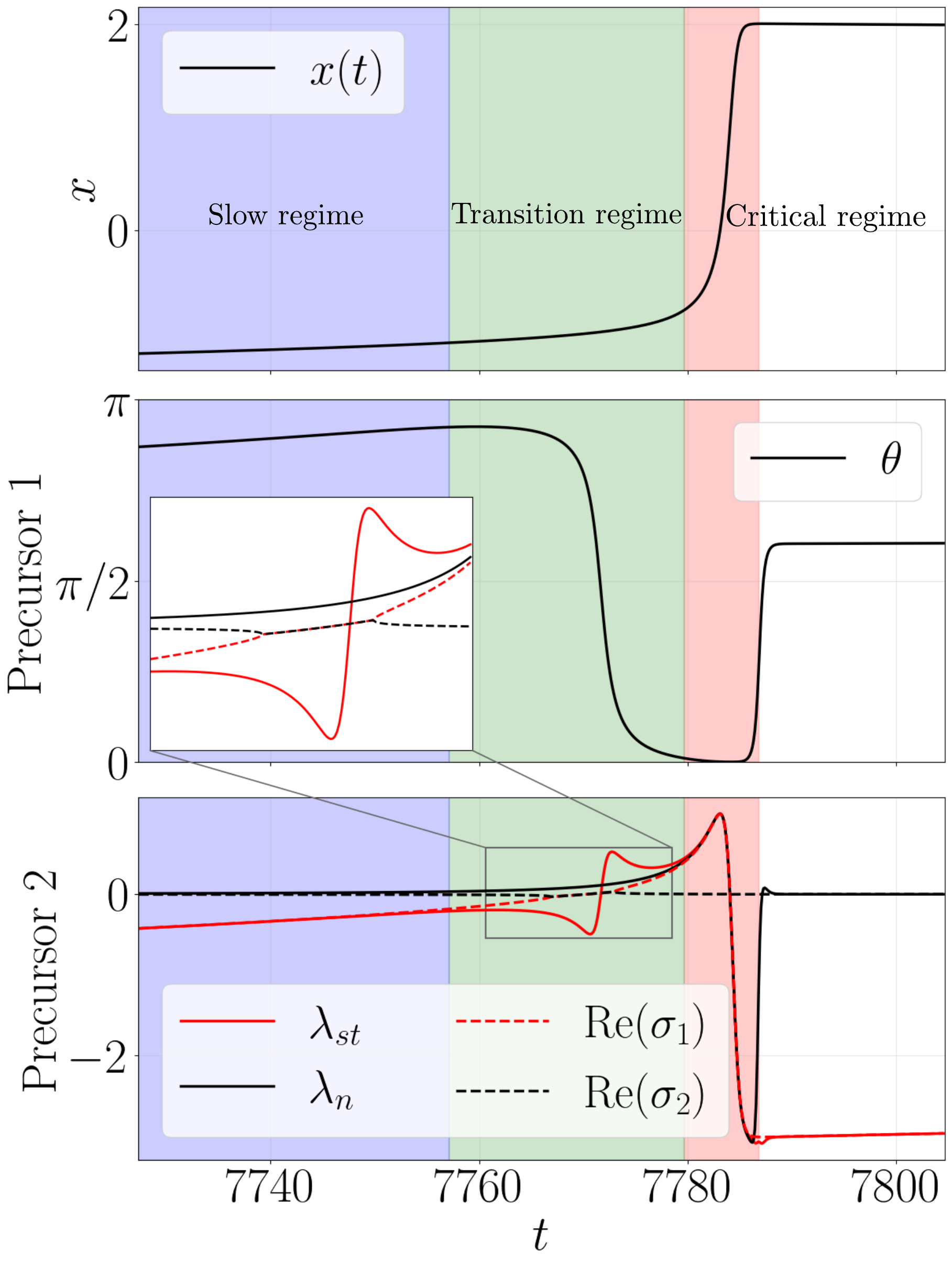}
    \caption{Van der Pol system. Top: evolution of the $x$ variable. Middle: evolution of the angle between the two CLVs $\theta$. Bottom: evolution of the ICLEs $\lambda_{st}$ and $\lambda_n$ and the real part of the eigenvalues. In the slow regime, the fast ICLE and the fast eigenvalue are equal, as shown in the bottom panel. In the transition regime, the fast ICLE and the fast eigenvalue decouple and the angle between CLVs quickly diminishes (both middle and bottom panels). In the critical regime, the adiabatic approximation holds and the CLVs become tangent to the fast eigenvector, as shown in the middle and bottom panels.}
    \label{icle_eig_vdp}
\end{figure}

When $x\approx-1$, the system is near a critical transition and the fast eigenvalue approaches $0$. Here, there is no clear time-scale separation between $x$ and $y$. In this range we have the transition regime, where the fast CLV detaches from the fast eigenvector. The eigenvalues become coupled and form a complex conjugate pair, {which induces a rotation of the fast CLV, as shown in Figure \ref{clv_vdp}.}
After {$x$ {becomes larger than} $-1$}, the system undergoes a critical transition, the fast eigenvalue becomes positive, and since the eigenvector remains stationary, the CLVs become tangent to it; as a consequence, the ICLEs converge to the value of the fast eigenvalue.

\begin{figure} 
    \centering
    \includegraphics[width=0.9\linewidth]{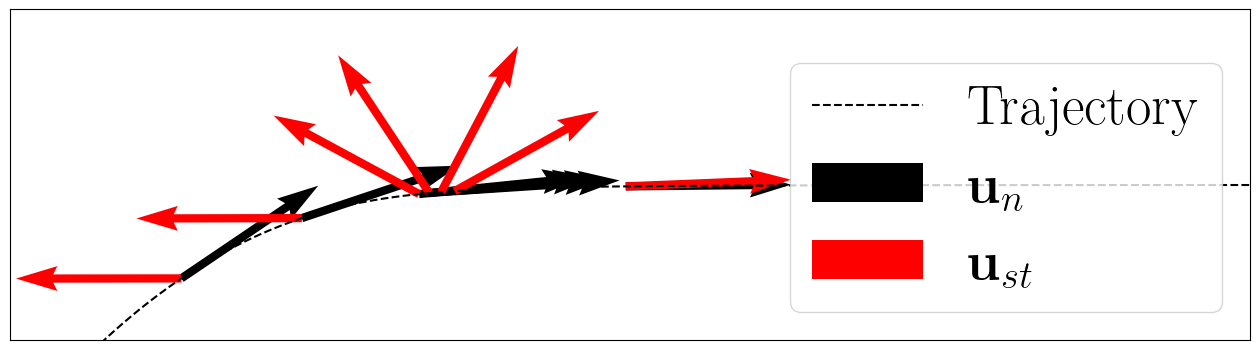}
    \caption{Van der Pol system, CLVs of the fast-slow Van der Pol oscillator at the fold bifurcation. Before the critical transition, the stable CLV rotates, indicating that the system is undergoing the transition regime.}
    \label{clv_vdp}
\end{figure}

\section{Numerical Results} \label{section:numerical results}
In this section, we perform numerical tests that validate the theory introduced in Sections \ref{section:clv_equation}, \ref{section: stability of clvs}, \ref{section:fast-slow systems} and \ref{sec: adiabatic condition and dynamical regimes}.

\subsection{Bistable Rössler}\label{sec:b_rossler}
We consider an explicitly fast-slow chaotic system \cite{rosslerChaoticHierarchy1983}

\begin{equation} \label{skew-rossler}
    \begin{aligned}
        \frac{dx}{dt} &= -y - a z, \\
        \frac{dy}{dt} &= x + b y, \\
        \varepsilon\frac{dz}{dt} &= (1 - z^2)(x - 1 + z) - \varepsilon z.
    \end{aligned}
\end{equation}

The first two equations describe an unstable planar spiral, while the third introduces a nonlinear term that governs a critical transition. The slow dynamics are equivalent to those of the Rössler system, whereas the fast dynamics exhibit a fold-type dynamic bifurcation. Consequently, the overall behaviour shares key qualitative features with the system analysed previously.
The system is integrated using a fourth-order explicit Runge–Kutta scheme with a time step of $\Delta t = 0.01$ and parameters $a = 0.95$, $b = 0.15$, and $\varepsilon = 0.03$. A reference trajectory of the system is shown in Figure~\ref{fig:b-rossler_attractor}.
Similarly to the classical Rössler system, the trajectory of the bistable Rössler oscillator spends most of its evolution on an expanding spiral, which acts as the slow manifold. In this case, however, the fast subsystem induces abrupt transitions between the two states $z = \pm 1$, similarly to the relaxation oscillations of the Van der Pol system. These transitions originate from a fold bifurcation in the fast dynamics. The portion of the attractor at $z = 1$ also acts as a slow attracting manifold on which the dynamics evolve; therefore, the transitions move the state from one attracting branch of the critical manifold to another.

\begin{figure}
    \centering
    \includegraphics[width=0.95\linewidth]{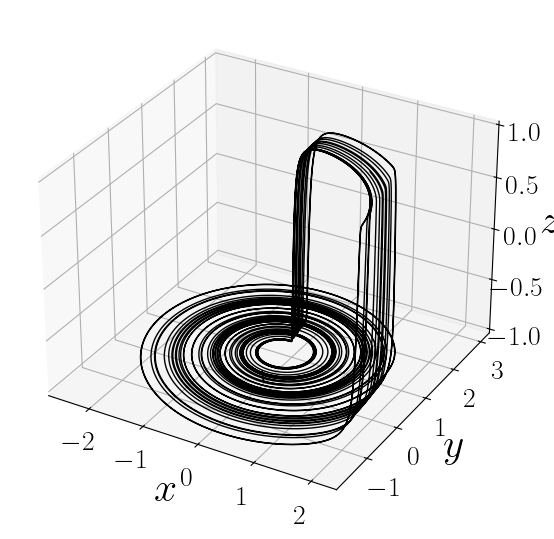}
    \caption{Bistable Rössler system chaotic attractor.}
    \label{fig:b-rossler_attractor}
\end{figure}

When the trajectory is on the slow manifold, the ICLE values correspond to the eigenvalues of the system's Jacobian (Figure \ref{fig:b-rossler_icle}), indicating that the fast CLV is tangent to the fast eigenvectors, exactly as described in Section \ref{dynamical_regimes}. This is the slow regime. {When the eigenvalue approaches zero}, the system undergoes the transition regime, the fast CLV is no longer equal to the fast eigenvector of the Jacobian and, thus, it is repelled towards the subspace spanned by the slow CLVs. This effect can also be appreciated if we analyse the stable ICLE: during the slow regime the ICLE is equal {to the fast eigenvalue, $\lambda_{st} \approx \mathrm{Re}(\sigma_f)$;} when the system undergoes the transition regime, the ICLE quickly moves to zero, implying that the fast CLV is transitioning to a subspace linked to a slower time scale.
After the crossing, the fast eigenvalue becomes positive and the system undergoes the critical regime. The $(1,2)$ fast-slow structure implies that the fast eigenvector will be approximately stationary and will therefore become a stable fixed point of the CLV dynamics. All the CLVs will become tangent to it, and the ICLEs will have the same values as the fast eigenvalue. These effects are shown in Figure \ref{fig:b-rossler_icle}.

\begin{figure}
    \centering
    \includegraphics[width=1\linewidth]{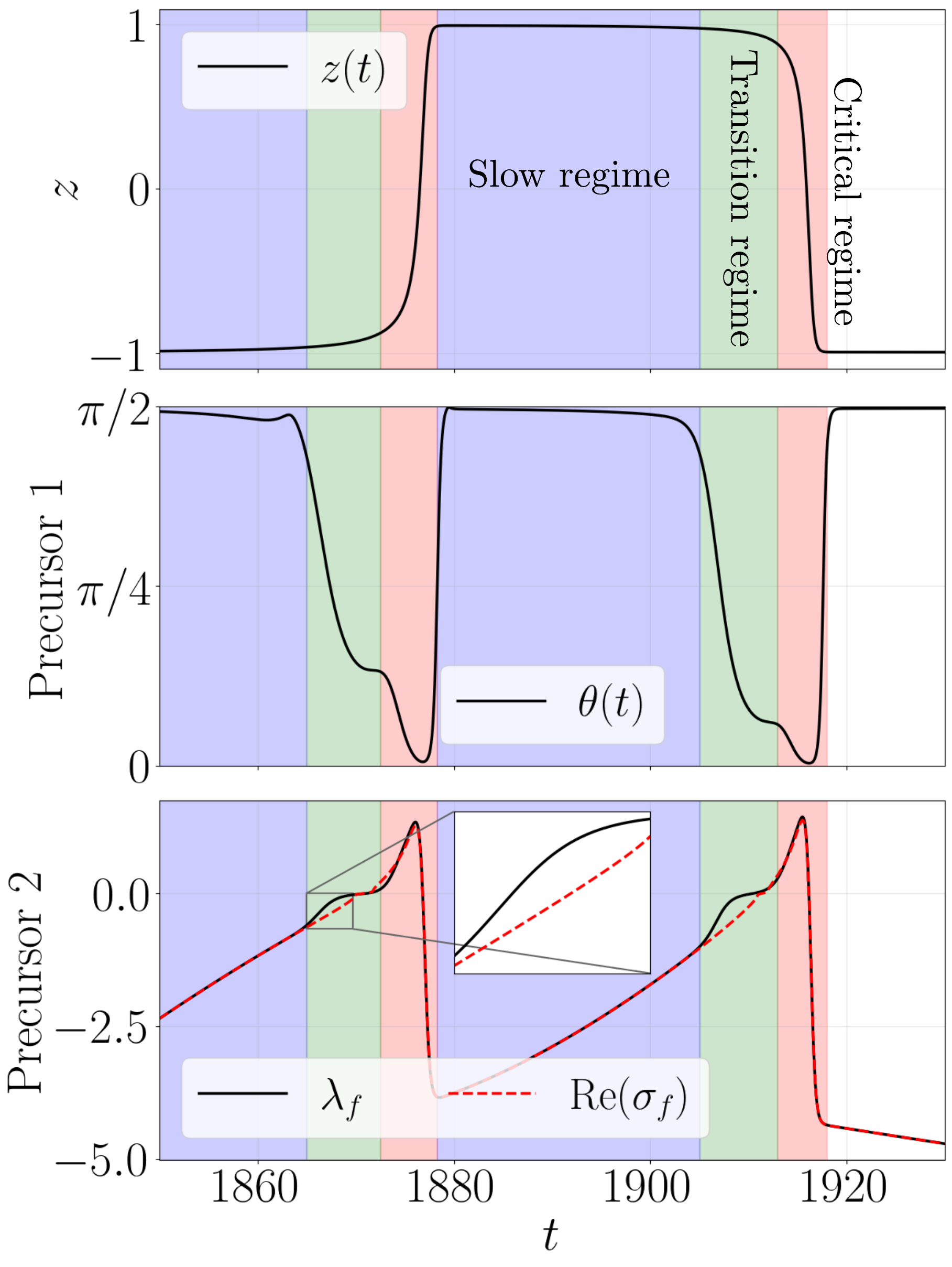}
    \caption{Bistable Rössler system. Top: time evolution of the $z(t)$ variable. Middle: time evolution of the cosine of the angle between the fast CLV and the subspace spanned by the slow CLVs. Bottom: time evolution of the ICLE and the fast eigenvalue. In the slow regime, the fast ICLE and the fast eigenvalue coincide. In the transition region, the ICLE and the eigenvalue decouple and the principal angle between fast and slow CLVs collapses rapidly. During the critical regime, the fast eigenvalue becomes larger than zero and triggers the critical transition, the CLVs become tangent along the eigendirection associated with the unstable fast eigenvalue, and the principal angle becomes zero.}
    \label{fig:b-rossler_icle}
\end{figure}

We analyse the statistics of precursors using a time series of length $100000$ time units with $297$ critical transitions. We use Precision, Recall, and F1-score to evaluate the prediction capabilities of our algorithms
\begin{align}
    P &= \frac{\text{TP}}{\text{TP} + \text{FP}} \quad \text{(Precision)} \\
    R &= \frac{\text{TP}}{\text{TP} + \text{FN}} \quad \text{(Recall)} \\
    F1 &= 2 \frac{P \times R}{P + R} \quad \text{(F1-score)},
\end{align}
{where TP, FP, FN are the numbers of \textit{true positive}, \textit{false positive}, and \textit{false negative} outcomes, respectively.}

{Using Precursor 1 with threshold $\alpha=0.94$ ($\approx20^\circ$ angle), we find $P = 1$, $R = 1$, and $F1 = 1$. The mean forewarning time (MFT) is $5.88$ time units and the standard deviation (SD) is $2.79$ time units. The maximum forewarning time is $12.03$ time units and the minimum forewarning time is $2.00$ time units.}

{Using Precursor 2 with threshold $\alpha = 0.2$, we obtain $P = 1.0$, $R = 1.0$, and $F1 = 1.0$. The MFT is $7.94$ time units and the SD is $3.18$ time units. The maximum and minimum forewarning times are $13.59$ and $3.33$ time units, respectively. Table \ref{table:stats} summarises the results and Figure \ref{fig:ee_stats} shows the time distribution of the forewarning time.}

\subsection{Coupled Fast-Slow FitzHugh--Nagumo Units}\label{sec:fhn}
We consider a ten-dimensional system of coupled FitzHugh--Nagumo units \cite{ansmannExtremeEventsExcitable2013}, described by

\begin{align}
\label{eq:fhn10}
    \begin{aligned}
        \varepsilon\dot{x}_i &= x_i (a - x_i)(x_i - 1) - y_i + k \sum_{j=1}^{n} A_{ij} (x_j - x_i), \\
        \dot{y}_i &= b_i x_i - c y_i,
    \end{aligned}
\end{align}

for $i = 1,...,n$, where \( x_i \) and \( y_i \) are the fast and slow state variables, respectively. The parameter \( \varepsilon > 0 \) is a small time-scale separation parameter, indicating that \( x_i \) evolves on a fast time scale relative to \( y_i \). The cubic term \( x_i(a - x_i)(x_i - 1) \) generates bistable or excitable behaviour depending on the choice of the parameter \( a \). The variable \( y_i \) serves as a slow recovery variable, linearly coupled to \( x_i \) through the parameters \( b_i \) and \( c \).
The coupling term \( k \sum_{j=1}^{n} A_{ij}(x_j - x_i) \) describes interactions between fast variables, where \( k \) is the coupling strength and \( A_{ij} \) is the connectivity matrix encoding the underlying coupling structure. Entries of \( A_{ij} \) indicate that node \( j \) influences node \( i \), while the term \( (x_j - x_i) \) ensures that coupling acts to reduce state differences among connected variables.

In the simulations we use $n = 5$, $a = -0.02651$, $c = 0.02$, $k = 0.00128$, $\varepsilon = 0.01$, and $b_i = 0.006 + 0.008 \frac{i-1}{n-1}$ for $i = 1,...,n$. The connectivity matrix $A$ is full, $A_{ij} = 1, \; \forall i,j$. These equations are integrated numerically using an explicit fourth-order Runge–Kutta integrator with a time step $\Delta t = 1$. A typical time series of the mean of the $n$ fast variables is shown in Figure \ref{fig:FHN_timseries}. This set of equations generates periodic motions with a rich structure. In particular, the mean of the fast variables exhibits different critical transitions that can be used as a benchmark to understand the generic behaviour of CLVs in systems of this kind.

\begin{figure}
    \centering
    \includegraphics[width=0.95\linewidth]{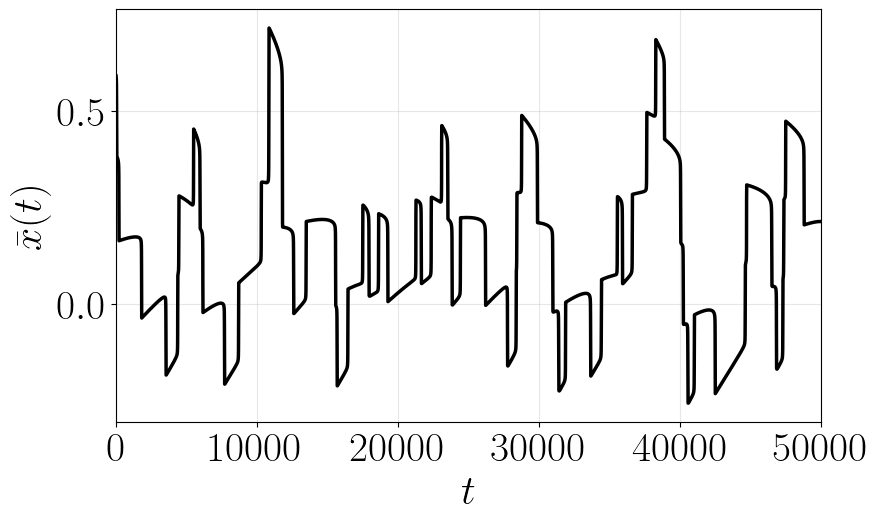}
    \caption{FitzHugh--Nagumo coupled oscillator system, time series of $\bar{x}(t) = \frac{1}{n} \sum_{i=1}^{n}x_i(t)$.}
    \label{fig:FHN_timseries}
\end{figure}

In Figure \ref{fig:FHN_icle}, the time evolution of the angles between fast and slow CLV subspaces is shown, together with the evolution of the sixth ICLE and eigenvalues.
{During the slow regime, the ICLE is equal to the real part of the sixth eigenvalue, confirming that fast CLVs become tangent to fast eigenvectors during the slow regime, as predicted in Section \ref{dynamical_regimes}. The largest eigenvalue is $\mathcal{O}(\varepsilon)$ in the slow regime, due to the fast-slow structure of system \eqref{eq:fhn10}. The angle between fast and slow CLV subspaces decreases slowly, similarly to the Van der Pol system studied in Section \ref{subsection: vanderpol}.}

{In the transition regime, before the critical transition, the sixth ICLE detaches from the sixth eigenvalue, indicating that the fast CLV decouples from the fast eigenvector during the transition regime, as predicted in Section \ref{dynamical_regimes}. In this case, the fast CLVs do not collapse onto the tangent space, but, similarly to the Van der Pol oscillator in Figure \ref{subsection: vanderpol}, the angle increases rapidly and, at the same time, the largest ICLE spikes rapidly. This indicates that some of the fast eigenvalues are coupling with some slow eigenvalues via complex conjugation and the fast CLV subspace is rotating as a consequence of Proposition \ref{proposition 3}. The precursor should be taken as the detachment of the CLV from the eigenvector during the transition regime.}

{In the critical regime, the state detaches from the slow manifold and the extreme event begins. The maximum eigenvalue of the Jacobian becomes positive and, due to a spectral gap, starts attracting the other CLVs. Under the adiabatic hypothesis, the CLVs become tangent to the eigenvector of the Jacobian with the maximum eigenvalue, which is exactly what is shown in Figure \ref{fig:FHN_icle}: the sixth ICLE and the maximum eigenvalue coincide during the critical transition, indicating that the sixth CLV is tangent to the associated eigenvector. This phenomenon generates tangencies, as predicted in Section \ref{dynamical_regimes}.}

{We evaluate the prediction performance of Precursor 2 for the FitzHugh--Nagumo system. We set $\alpha = 0.18$ and find $P = R = F1 = 1$ with $\text{MFT} = 31.89$ time units and $\text{SD}=13.37$. The maximum forewarning time is $47$ time units and the minimum forewarning time is $3$ time units. Table \ref{table:stats} summarises the precursor performance and Figure \ref{fig:ee_stats} shows the time distribution of predictions.}
 {Precursor 1 is not applicable to this case because the angle between fast and slow CLVs does not collapse during the transition regime due to a complex conjugacy between eigenvalues (this case belongs to Proposition \ref{proposition 3}. 

\begin{figure}
    \centering
    \includegraphics[width=0.95\linewidth]{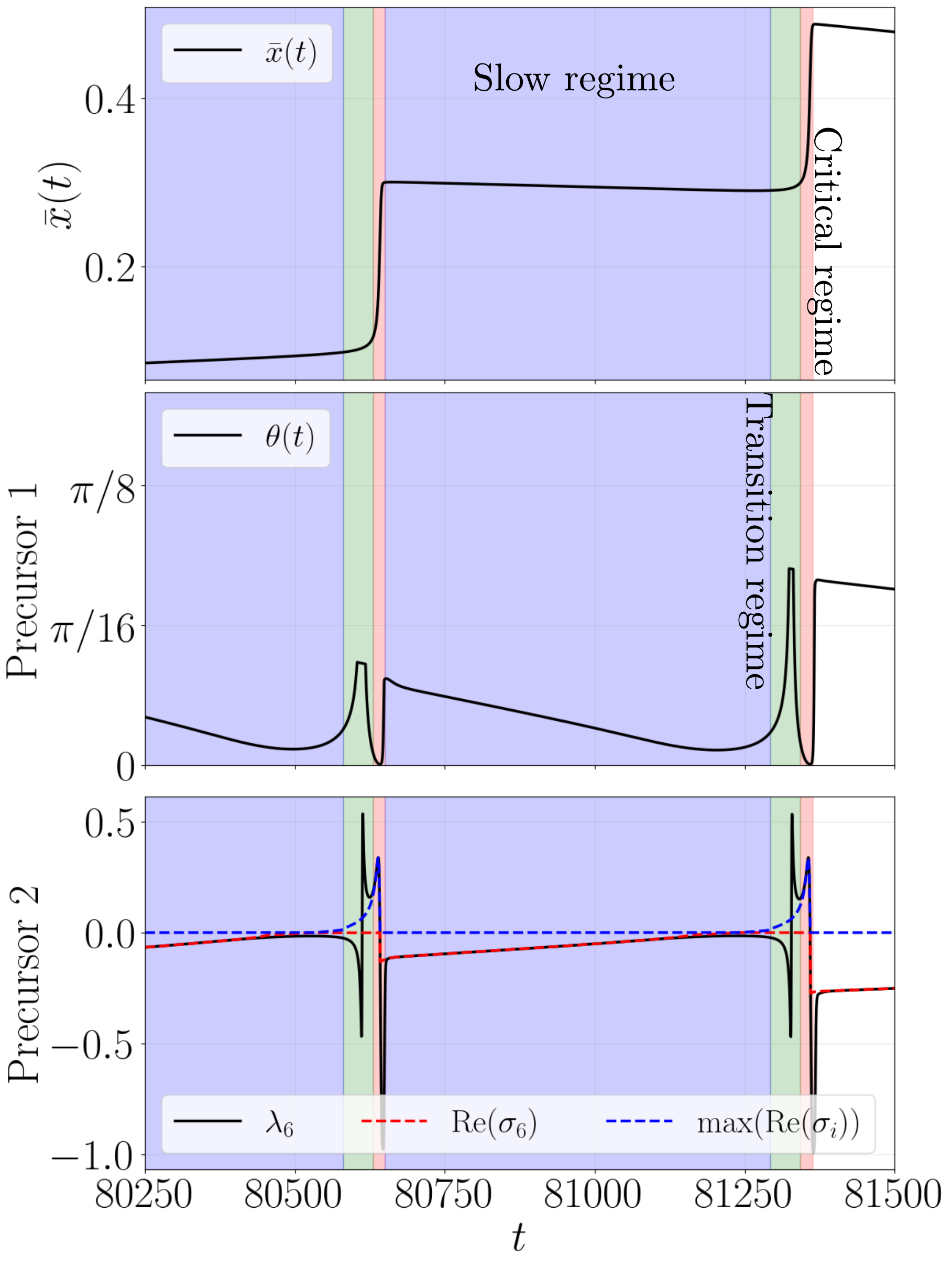}
    \caption{FitzHugh--Nagumo oscillator. First panel from above: average of the $n$ fast states $\bar{x}(t) = \frac{1}{n} \sum_{i=1}^{n}x_i(t)$. Second panel: smallest principal angle between the fast CLV subspace and the slow CLV subspace. Third panel: minimum ICLE and minimum eigenvalue. Fourth panel: maximum ICLE and maximum eigenvalue. For all critical transitions shown in the top panel, all three regimes introduced in Section \ref{dynamical_regimes} are identifiable.}
    \label{fig:FHN_icle}
\end{figure}

\subsection{Multiscale Lorenz 96 with Critical Transitions}\label{sec:l96}
We consider a modified multiscale Lorenz 96 \cite{lorenzPredictabilityProblemPartly1995, carluLyapunovAnalysisMultiscale2019} system in which we add an equation to allow the system to transition from a ``low temperature'' state to a ``high temperature'' state.

\begin{equation} \label{eq: multiscale_lorenz96}
    \begin{aligned}
        \frac{dX_k}{dt}
         &= X_{k-1}\bigl(X_{k+1}-X_{k-2}\bigr) - X_k \\
         &- \sum_{j=1}^{J} Y_{j,k} + F_X + \alpha \, Z, \\
         \varepsilon \frac{dY_{j,k}}{dt}
         &= b\, Y_{j+1,k}\bigl(Y_{j-1,k}-Y_{j+2,k}\bigr)
         - Y_{j,k} \\
         &+ X_k + \frac{F_Y}{b} + \frac{\beta}{b} \, Z,\\
        \eta \frac{dZ}{dt}
        &= \left(1 - Z^2 \right) \left( Z + \gamma E(X,Y) -1 \right) - \eta \, Z.
    \end{aligned}
\end{equation}

Cyclic boundary conditions are imposed:

\begin{equation}
X_{k+K} = X_k, \quad
Y_{j+J,k} = Y_{j,k}, \quad
Y_{j,k+K} = Y_{j,k}.
\end{equation}

Here, \(X_k\) (\(k = 1, \dots, K\)) are the slow variables, each coupled to \(J\) fast variables \(Y_{j,k}\) (\(j = 1, \dots, J\)). The scalar \(Z\) is a global-scale variable, which may be interpreted as a large-scale climatic or planetary-scale feedback process that interacts with the lower-scale dynamics.
The constants \(F_X\) and \(F_Y\) are external forcings acting on the slow and fast variables, respectively. The parameters \(\alpha\) and \(\beta\) control the coupling strength between the global variable \(Z\) and the slow and fast components. The coefficient \(b\) scales the amplitude and nonlinearity of the fast subsystem, while \(\varepsilon \ll 1\) and \(\eta \ll 1\) set the time-scale separation of the fast (\(Y\)) and global (\(Z\)) variables relative to the slow (\(X\)) dynamics.
The nonlinear feedback in \eqref{eq: multiscale_lorenz96} depends on the energy \cite{carluLyapunovAnalysisMultiscale2019}

\begin{equation}
E(X,Y) = \sum_{k=1}^{K} X_k^2
       + \sum_{k=1}^{K}\sum_{j=1}^{J} Y_{j,k}^2.
\end{equation}

When the system reaches a certain cutoff value of $E(X,Y)$, it transitions to a higher energy state. The parameter \(\gamma\) modulates how strongly the global mode responds to the combined energy of the lower-scale processes.
The test case has $K = 36$, $J = 10$, $b = 10$, $\alpha = 4$, $\beta = 2$, $\gamma = 1/445$, $F_X = 14$, $F_Y = 8$, $\varepsilon = 0.1$, and $\eta = 0.01$. The system is integrated with an explicit fourth-order Runge–Kutta integrator with a time step $\Delta t = 0.001$.

\begin{figure}
    \centering
    \includegraphics[width=1\linewidth]{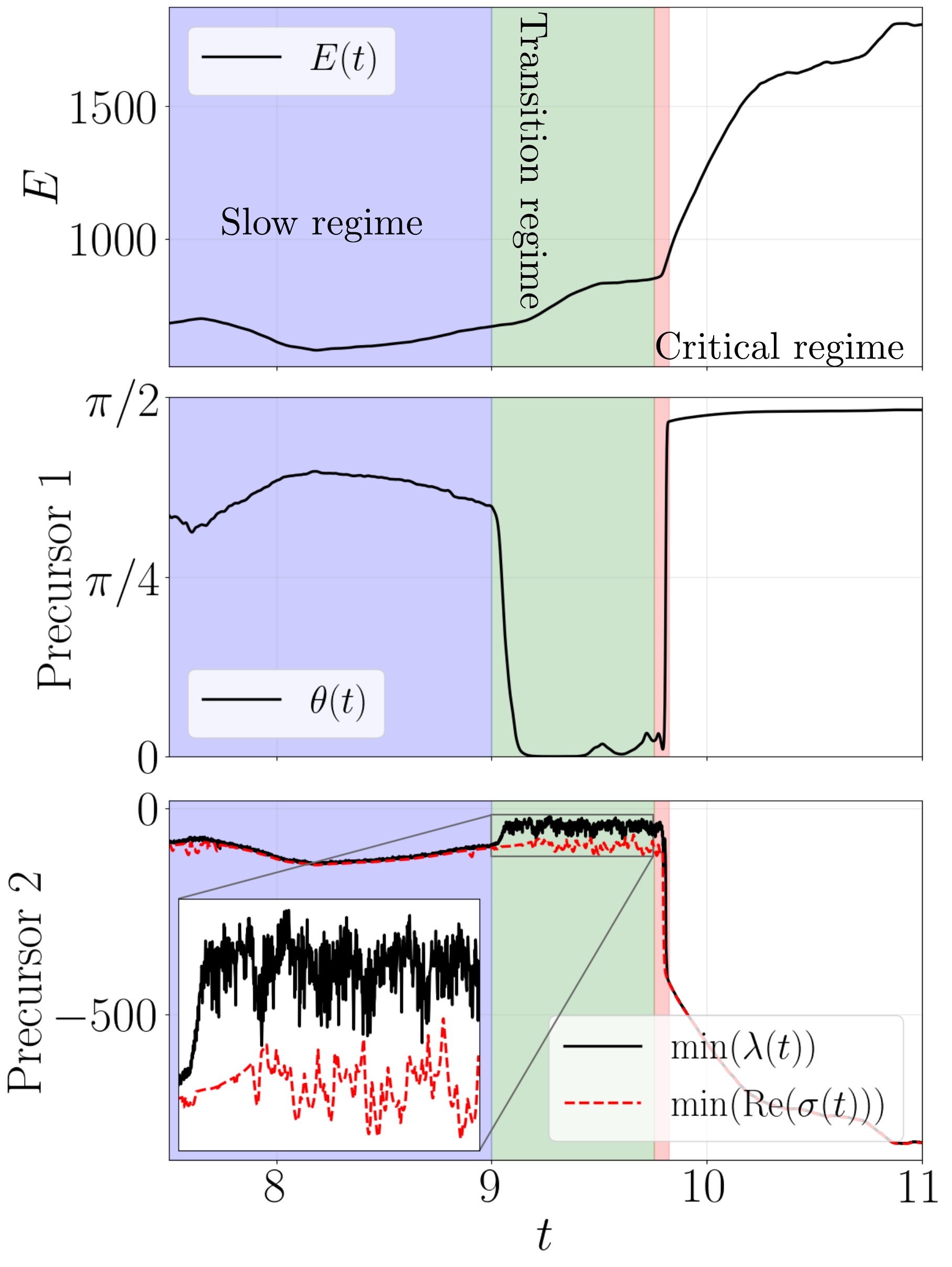}
    \caption{Multiscale Lorenz 96 system. Top: time series of the energy of the system. Middle: angle between fast and slow CLV subspaces. Bottom: time series of the fastest ICLE and minimum eigenvalue. At time $t \approx 9.5$, system \eqref{eq: multiscale_lorenz96} undergoes a large critical transition from its low-energy state to its high-energy state. Before the transition, the principal angle collapses to zero and the minimum eigenvalue and minimum ICLE decouple, indicating that the system is undergoing the transition regime.}
    \label{fig:L96_icle}
\end{figure}

As described in Section \ref{dynamical_regimes}, during the slow regime the ICLE connected to the fastest time scale is equal to the smallest eigenvalue of the Jacobian of the dynamical system (Figure \ref{fig:L96_icle}). When this eigenvalue becomes large enough, the instantaneous time scale associated with the fast global mode becomes comparable to the characteristic time scale of the $Y$ variables. Hence, the smallest ICLE and the smallest eigenvalue decouple, indicating that the fastest CLV is no longer equal to the fastest eigenvector. The fastest CLV is then repelled toward the slower CLV subspace, thereby generating a tangency.

{In this case, the critical transition occurs at $t \approx 9.75$ as shown in Figure \ref{fig:L96_icle}; both CLV tangencies and ICLE–eigenvalue decouplings act as precursors of the extreme events during the transition regime.}
{To test the performance of Precursor 1 and Precursor 2, we run an ensemble of $100$ trajectories of Equation \eqref{eq: multiscale_lorenz96}, with a total of $90$ critical transitions, with $K = 20$, $J = 3$, $b = 10$, $\alpha = 4$, $\beta = 2$, $\gamma = 1/340$, $F_X = 14$, $F_Y = 8$, $\varepsilon = 0.5$, and $\eta = 0.005$, using an explicit fourth-order Runge–Kutta integrator with time step $\Delta t = 0.001$.}
{For Precursor 1, we set $\alpha = 0.985$ and achieve $P = 0.989$, $R = 1$, $F1 = 0.994$ with $\text{MFT} = 0.278$ time units and $\text{SD} = 0.146$ time units. The minimum forewarning time is $0.064$ time units and the maximum forewarning time is $0.858$ time units.
}
{For Precursor 2, we set $\alpha = 67$ and find $P = 1$, $R = 1$, $F1 = 1$ with $\text{MFT} = 0.123$ time units and $\text{SD} = 0.111$ time units. The minimum forewarning time is $0.043$ time units and the maximum forewarning time is $0.870$ time units.
Table \ref{table:stats} summarises the precursors' performance. Figure \ref{fig:ee_stats} shows the prediction time distributions.}

\begin{table*}
\caption{Comparison of algorithms across dynamical systems.}
\centering
\label{table:stats}
\begin{tabular}{|c|c|c|c|}
\hline
 & Bistable Rössler & FitzHugh--Nagumo & Multiscale Lorenz 96 \\
\hline
Precursor 1 & $\text{F-score = 1}$ &{NA} &  $\text{F-score = 0.99}$\\
\hline
Precursor 2 & $\text{F-score = 1}$  & $\text{F-score = 1}$ & $\text{F-score = 1}$ \\
\hline
\end{tabular}
\end{table*}

\begin{figure}
    \centering
    \includegraphics[width=1\linewidth]{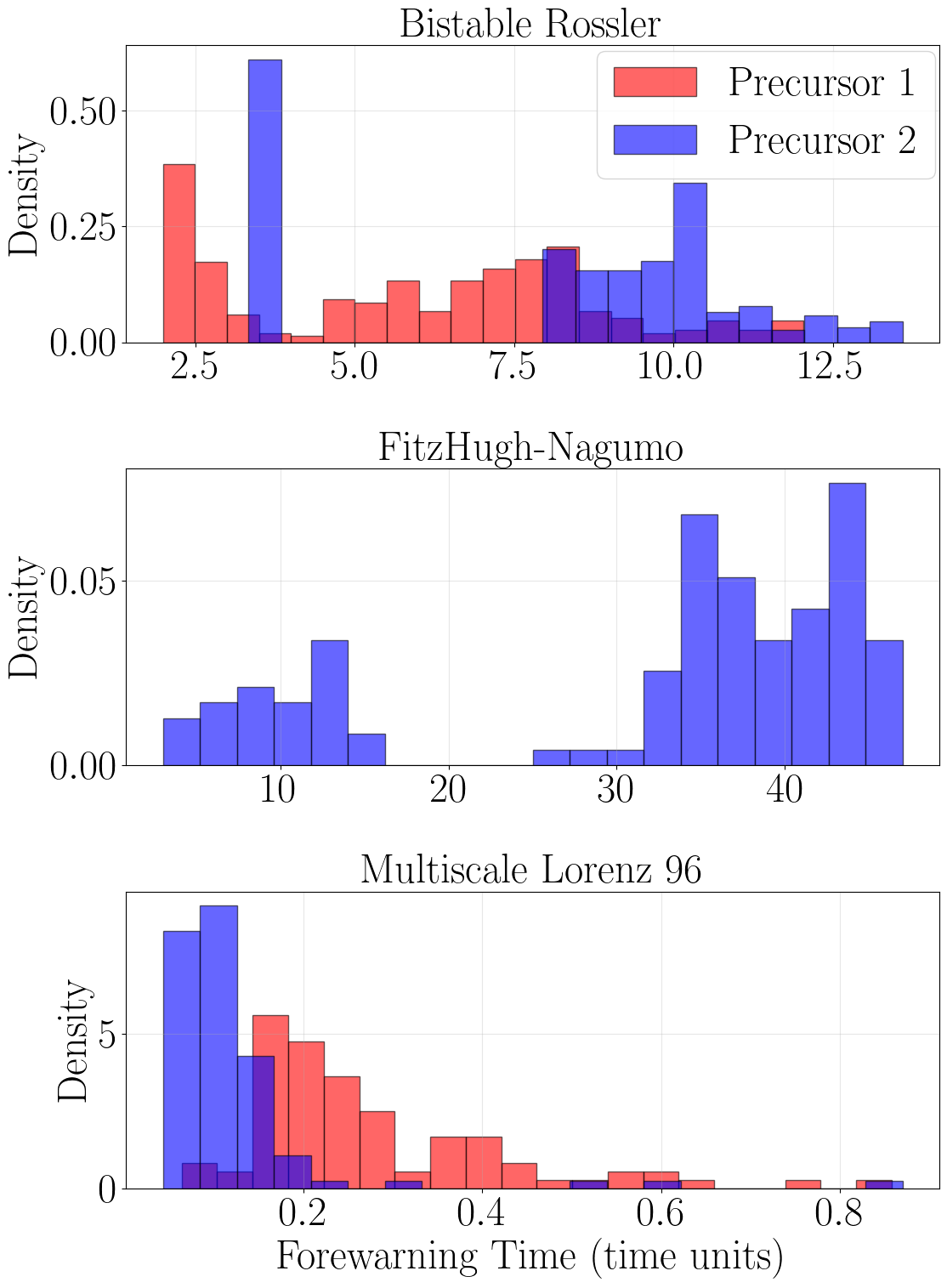}
    \caption{Forewarning time of the precursors of extreme events. Upper panel: Bistable Rössler. Centre panel: FitzHugh--Nagumo. Lower panel: Multiscale Lorenz 96. }
    \label{fig:ee_stats}
\end{figure}
 
\section{Conclusions}\label{section:conclusion}
We developed a theory based on fast-slow dynamical systems to explain and predict the occurrence of extreme events, of which critical transitions form a subset. Building upon the evolution equation for covariant Lyapunov vectors (CLVs) as a projected flow on the unit sphere, we derived a theoretical framework in which invariant sets such as fixed points govern CLV behaviour, and CLV tangency along a dominant eigendirection is interpreted as relaxation toward an attracting fixed point. In fast-slow systems, we identified a cascade of dynamical regimes preceding extreme events and critical transitions: (i) a {\it slow regime}, in which fast CLVs coincide with the fast eigenvectors and remain transversal to the slow subspace, which coincides with the tangent space of the slow manifold; (ii) a {\it transition regime}, in which fast eigenvalues approach zero while the fast CLVs detach from the fast eigenvectors through either rotation or attraction toward the slow tangent space; and (iii) a {\it critical regime}, in which a dominant positive eigenvalue attracts multiple CLVs, producing tangencies and breaking the transversality between fast and slow subspaces. Building on this theory, we proposed two precursors: (i) monitoring the principal angles between fast and slow CLV subspaces, and (ii) tracking the decoupling between fast instantaneous covariant Lyapunov exponents and Jacobian fast eigenvalues. We numerically validated the theory and precursors on systems ranging from low to higher dimensions, including a Van der Pol oscillator, a bistable Rössler system, coupled FitzHugh-Nagumo units, and a modified multiscale Lorenz-96 system, where the proposed precursors consistently forewarn the occurrence of extreme events and critical transitions with 100\% precision and recall. Future work will focus on extending the framework to higher-dimensional multiscale systems. This work opens opportunities for predicting extreme events with theoretically justified precursors in the tangent space.

\begin{acknowledgments}
The authors acknowledge funding from the Starting Grant of Fondazione Compagnia di San Paolo (CSP).
L.M. also acknowledges funding from the ERC Starting Grant PhyCo n.~949388.
\end{acknowledgments}

%\nocite{*}
\bibliographystyle{apsrev4-2}
\bibliography{refs}% Produces the bibliography via BibTeX.

\end{document}